\documentclass[runningheads]{llncs}
\usepackage[T1]{fontenc}
\usepackage{lmodern} 

\usepackage{graphicx} 
\usepackage{amsmath,amssymb} 
\usepackage{hyperref} 
\usepackage{thmtools,cleveref} 
\usepackage{thm-restate} 
\usepackage[shortlabels]{enumitem}
\usepackage{cite}
\usepackage{wrapfig} 

\usepackage{xcolor} 
\usepackage{tikz} 
\usetikzlibrary{arrows,positioning,shapes,decorations,automata,backgrounds,petri,fit,calc,patterns} 

\makeatletter
\providecommand{\tikz@nonactivesemicolon}{;}
\providecommand{\tikz@nonactivecolon}{:}
\providecommand{\tikz@deactivatethings}{}%
\providecommand{\tikz@deactivatthings}{\tikz@deactivatethings}%
\makeatother

\usepackage{array} 
\usepackage{lipsum} 
\usepackage{etoolbox} 
\usepackage{soul} 
\usepackage{mathtools} 
\usepackage{enumitem} 
\newlist{inlineenum}{enumerate*}{1} 
\setlist*[inlineenum]{mode=unboxed,label=\alph*} 
\usepackage{blindtext} 
\usepackage{multicol} 
\usepackage{mdframed}
\usepackage[misc,geometry]{ifsym}

\usepackage[caption=false]{subfig} 
\usepackage{algorithm}

\usepackage[noend]{algpseudocode}

\spnewtheorem{note}{Observation}{\bfseries}{\rmfamily}



\spnewtheorem{claim}{Claim}{\bfseries}{\itshape}

\newcommand{\commentout}[1]{}

\newcommand{\class}[1]{\mathbb{#1}}

\newcommand{\caclassshort}[3]{\mathbb{CA}_{\scriptstyle{#2,#3}}}

\newcommand{\caone}[2]{\mathbb{CA}_{\scriptstyle{#1,#2}}}  

\newcommand{\caclassglider}[2]{\mathbb{CA}^{\substack{\scriptscriptstyle\textup{Glider}}}_{\scriptstyle {#2,#1}}}

\newcommand{\minus}{\scalebox{0.65}[1]{$-$}}

\newcommand{\miniminus}{\scalebox{0.35}[1]{$-$}}

\newcommand{\mminus}{\scalebox{0.75}[0.75]{$-$}}

\newcommand{\FN}{\textsl{FN}}
\newcommand{\FNs}{\textsl{FNs}}

\newcommand{\term}[1]{\text{\emph{#1}}}
\newcommand{\mathterm}[1]{\emph{$#1$}}

\newcommand{\orbitc}[1]{\mathcal{O}(#1)}

\newcommand{\mybot}{\bot}
\DeclareTextFontCommand{\word}{\emph}

\DeclareTextFontCommand{\proceduresfont}{\normalfont\scshape}

\newcommand{\mainproc}{\proceduresfont{GtoR}}

\newcommand{\g}[2] {{\mathsf{g}}{(#1,#2)}}
\newcommand{\plaing}{{\mathsf{g}}}

\newcommand{\gfig}[2] {{\mathsf{g}}_{(#1,#2)}}
\newcommand{\Greg}{\mathrm{\mathsf{\boldsymbol{G}}}}

\newcommand{\G}{\text{\small{${\Greg}$}}}

\newcommand{\procmaxx}{\text{\textsc{Max}}}

\newcommand{\extend}{\text{\textsc{Extend}}}

\newcommand{\glider}[1]{\mathsf{g}_{\scriptstyle#1}}
\newcommand{\gglider}{\mathsf{g}}

\makeatletter
\DeclareRobustCommand{\fontitoneC}[1]{%
  \text{\normalfont
    \@ifpackageloaded{lmodern}{\bfseries\scshape #1}{\bfseries\MakeUppercase{#1}}%
  }%
}
\makeatother

\newcommand{\SigmaWindow}{\ensuremath{\Sigma^{\scriptstyle 2r{+}1}}}

\newcommand{\SigmaZWindow}{\ensuremath{\Sigma^{\scriptstyle{\mathbb{Z}}}}}

\newcommand{\crunss}{\text{\normalfont \texttt{{Cruns}}}}

\newcommand{\crunso}[1]{\crunss(#1)}

\newcommand{\itempos}[2]{#1_{\scriptstyle #2}}

\newcommand{\itemneg}[2]{#1_{\scriptstyle {-#2}}}

\newcommand{\sig}[1]{{x}_{\scriptstyle{#1}}}

\newcommand{\apos}[1]{\itempos{a}{#1}}
\newcommand{\bpos}[1]{\itempos{b}{#1}}

\newcommand{\aneg}[1]{\itemneg{a}{#1}}

\newcommand{\mplus}{\scalebox{0.5}[0.5]{\boldmath$+$}}

\newcommand{\miniplus}{\raisebox{0.2ex}{\mplus}}

\newcommand{\wword}{word}

\newcommand{\macrolang}{language}
\newcommand{\macrolangs}{languages}

\newcommand{\fLang}{finitary \macrolang}

\newcommand{\infLang}{infinitary \macrolang}
\newcommand{\infLangs}{infinitary \macrolangs}

\newcommand{\biInfLang}{bi-\infLang}

\newcommand{\caclassshortREG}[4]{\mathbb{CA}^{\substack{\texttt{\scriptsize{#4}}}}_{#2,#3}}

\newcommand{\whitebox}{\tikz[baseline=0ex]\filldraw[fill=white!30,draw=black] (0,0) rectangle (1.5ex,1.5ex);}

\newcommand{\smallwwhitebox}{\tikz[baseline=0ex]\filldraw[fill=white!30,draw=white] (0,0) rectangle (1ex,1ex);}

\newcommand{\abox}{\tikz[baseline=0ex]\filldraw[fill=gray!30,draw=black] (0,0) rectangle (1.5ex,1.5ex);} 

\newcommand{\bbox}{\tikz[baseline=0ex]\filldraw[fill=darkgray!50,draw=black] (0,0) rectangle (1.5ex,1.5ex);} 

\newcommand{\smallbbox}{\tikz[baseline=0ex]\filldraw[fill=darkgray!50,draw=black] (0,0) rectangle (1ex,1ex);} 

\newcommand{\cbox}{\tikz[baseline=0ex]{\filldraw[fill=darkgray!60,draw=black] (0,0) rectangle (1.5ex,1.5ex);
\filldraw[pattern=north west lines, pattern color=white] (0,0) rectangle (1.5ex,1.5ex);

\filldraw[pattern=north east lines, pattern color=white] (0,0) rectangle (1.5ex,1.5ex);}
}

\newcommand{\smallcbox}{\tikz[baseline=0ex]{\filldraw[fill=darkgray!60,draw=black] (0,0) rectangle (1ex,1ex);
\filldraw[pattern=north west lines, pattern color=white] (0,0) rectangle (1ex,1ex);

\filldraw[pattern=north east lines, pattern color=white] (0,0) rectangle (1ex,1ex);}
}

\newcommand{\fillc}[3]{
\fill[fill=darkgray!60] (#1,#2) rectangle ++(#3, 1);
\fill[pattern=north east lines, pattern color=white](#1,#2) rectangle ++(#3, 1);
\fill[pattern=north west lines, pattern color=white](#1,#2) rectangle ++(#3, 1);
}

\newcommand{\pad}{{\mathsf{pad}}_{\bot}} 
\newcommand{\ppad}{{\mathsf{pad}}_{\bot}} 

\newcommand{\termsfont}[1]{\ensuremath{\textrm{#1}}}
\newcommand{\supp}{\termsfont{supp}}

\newcommand{\intervalInt}{\textrm{int}}
\newcommand{\interval}[1]{\textrm{int}(#1)}

\newcommand{\width}[1]{\textrm{width}(#1)}

\newcommand{\dom}{\textrm{\textbf{dom}}}
\newcommand{\sign}{\textrm{\textbf{sign}}}

\newcommand{\LocAt}[2]{\mathcal{N}^{\substack{#2}}_{#1}}

\newcommand{\whiteArrowback}[4]{
\foreach \dx/\dy in {-0.01/0, 0.01/0, 0/-0.01, 0/0.01, -0.02/0, 0.02/0, 0/-0.02, 0/0.02} {
           \draw[white, very thick, ->] ($(#1,#2)+(\dx,\dy)$) -- ($(#3,#4)+(\dx,\dy)$);
        }
        }

\newcommand{\blackonwhite}[3]{
\foreach \dx/\dy in {-0.01/0, 0.01/0, 0/-0.01, 0/0.01} {
          \node at ($(#1,#2)+(\dx,\dy)$) {\Huge \textcolor{white}{$\boldsymbol{#3}$}};
        }
        \node at (#1,#2) {\Huge \textcolor{black}{$\boldsymbol{#3}$}};
        }

\newcommand{\blackonwhiteScaled}[4]{
\foreach \dx/\dy in {-0.01/0, 0.01/0, 0/-0.01, 0.01/0, -0.015/0, 0.015/0, 0/-0.015,0/0.015, 0/0.02,-0.02/0, 0.02/0, 0/-0.02, 0/0.02, -0.01/-0.01,0.01/0.01, -0.01/-0.01, 0.01/0.01, 0/0.03,-0.03/0, 0.03/0, 0/-0.03, 0/0.03} {
          \node at ($(#1,#2)+(\dx,\dy)$) {#4 \textcolor{white}{$\boldsymbol{#3}$}};
        }
        \node at (#1,#2) {#4 \textcolor{black}{$\boldsymbol{#3}$}};
        }

\newcommand{\blackonwhiteScaledOne}[5]{
\foreach \dx/\dy in {-0.01/0, 0.01/0, 0/-0.01, 0.01/0, -0.015/0, 0.015/0, 0/-0.015,0/0.015, 0/0.02,-0.02/0, 0.02/0, 0/-0.02, 0/0.02, -0.01/-0.01,0.01/0.01, -0.01/-0.01, 0.01/0.01} {
          \node at ($(#1,#2)+(\dx,\dy)$) {#4 \textcolor{white}{$\boldsymbol{#5}$}};
        }
        \node at (#1,#2) {#4 \textcolor{black}{$\boldsymbol{#3}$}};
        }

\newcommand{\blackonwhiteArrow}[6]{
\foreach \dx/\dy in {-0.01/0, 0.01/0, 0/-0.01, 0/0.02,-0.02/0, 0.02/0, 0/-0.02, 0/0.02} {
          \draw[#5 ,white!60,->]
        ($(#1,#2)+(\dx,\dy)$) -- node[below =10pt, yshift=4pt, xshift=10pt] {~} ($(#3,#4)+(\dx,\dy)$);
        }
        \draw[#5 ,#6,->]
        (#1, #2) -- node[below =10pt, yshift=4pt, xshift=10pt] {~} (#3, #4);
        }

\newcommand{\overlapuline}[3][0.2ex]{
  \tikz[baseline]{
    \node[inner sep=0, outer sep=0, anchor=base west] (t) {#3};
    \draw[line width=#2] ([yshift=-#1]t.base west) -- ([yshift=-#1]t.base east);
  }%
}

\makeatletter
\let\c@lemma=\c@theorem
\let\c@corollary=\c@theorem
\let\c@proposition=\c@theorem
\let\c@claim=\c@theorem
\let\c@definition=\c@theorem
\let\c@example=\c@theorem
\let\c@remark=\c@theorem
\let\c@note=\c@theorem
\makeatother

\usepackage{orcidlink}

\title{Atomic Gliders and Cellular Automata as Language Generators}

\titlerunning{Atomic Gliders and CA as Language Generators}

\author{\normalsize{ 
Dana Fisman\orcidlink{0000-0002-6015-4170}\inst{1} \and 
Noa Izsak\orcidlink{0009-0004-1333-2490}\inst{2,1}\thanks{\Letter~noa.izsak@cispa.de}
}}

\authorrunning{D. Fisman and N. Izsak}
\institute{Ben Gurion University, Beer-Sheva, Israel 
\and
CISPA Helmholtz Center for Information Security, Saarbr\"ucken, Germany}

\begin{document}
\vspace{-1mm}

\maketitle

\vspace{-6mm}

\begin{abstract}
Cellular automata (CA) are well-studied models of decentralized parallel computation, known for their ability to exhibit complex global behavior from simple local rules. 
While their dynamics have been widely explored through simulations, a formal treatment of CA as genuine \emph{language generators} remains underdeveloped.
We formalize CA-expressible languages as sets of finite words obtained by projecting the non-quiescent segments of
configurations reachable by one-dimensional, deterministic, synchronous CA over bi-infinite grids.
These languages are defined with respect to sets of initial configurations specified by a regular language as in \emph{regular model checking}.
To capture structured dynamics, we propose a \emph{glider-based generative semantics} for CA.
Inspired by the classical notion of gliders, we define a glider as a one-cell entity carrying a symbol in a certain velocity under well defined interaction semantics. 
We show that despite the regularity of the initial configurations and the locality of the transition rules, the resulting languages can exhibit non-regular and even non-context-free structure.
This positions regular-initialized CA languages as a surprisingly rich computational model, with potential applications in the formal analysis of linearly ordered MAS.
\end{abstract}

\section{Introduction}
Multi-agent systems (MAS) arranged in a linear topology, consisting of identical agents (or processes) that operate in parallel, arise naturally in many applications. 
In these systems, agents are typically organized in a line or a ring. 
The agents evolve in synchronous rounds; at each step, all agents update their state simultaneously, each according to the state of its finite neighborhood (including itself).
Analyzing and verifying such systems is a central challenge in formal methods. 
For a single finite-state transition system, classical verification techniques apply directly since the system has an explicit finite-state space. 
\looseness=-1
In contrast, in a multi-agent system, all agents implement the same protocol, $P$.

Executing $P$ with $n$ indistinguishable agents yields a finite system, $P^n$. The resulting MAS induces an infinite family $\{ P^n\}_{n\in\mathbb{N}}$, known as a \emph{parametrized system}, as the number of agents is unbounded. 
Research on parametrized systems includes works on verification~\cite{AbdullaBJN99,EmersonK00,Bloem2016,AbdullaST18,JaberJW0S20} (to name a few), synthesis~\cite{JacobsB14,khalimov2013towards,lazic2018synthesis} and learning~\cite{fisman2024learning,fribourg1997reachability,ma2019i4}.

In their seminal work~\cite{KestenMMPS97}, Kesten et al. proposed the use of \emph{rich assertional languages} to symbolically capture the unbounded nature of the system families induced by such systems. 
This idea gave rise to the method now known as \term{regular model checking} (RMC), which has since developed into a well-established verification technique~\cite{Bouajjani2000,Abdulla04,Bouajjani04,fribourg1997reachability,wolper1998verifying,boigelot1999symbolic,jonsson2000transitive}.

In RMC, \emph{sets of configurations} are represented by regular languages over a finite alphabet, while \emph{transition relations} are expressed as regular relations, typically implemented by finite-state transducers. 
Using regular languages, one can succinctly describe the set of initial configurations for systems of all sizes. 
The alphabet specifies the local state-space of a process, and a word of length $n$ encodes the states of $n$ processes arranged along a line or a ring, with position $i$ corresponding to process $i$. 
For example, the regular language $10^*$ compactly denotes infinitely many systems, one for each $n\,{\in}\,\mathbb{N}$, where the leftmost process is in state $1$ (signifying that it holds a token), while all processes to its right are in state $0$ (not holding the token).

In each synchronous step, \emph{all} processes update their states simultaneously, based on their current state and the states of finitely many neighbors (e.g., their immediate left and right). 
In the case of the \emph{token-passing protocol}~\cite{le1977distributed,emerson1995reasoning}, the relation specifies that if a process does not currently hold a token but its left neighbor does, then in the next step the process acquires the token while the neighbor relinquishes it. 
The RMC framework has proven to be highly effective for automatic verification of \emph{parameterized systems}, supporting the analysis of a wide range of protocols~\cite{neider2013regular,Abdulla04,Bouajjani2000,abdulla2012regular,esparza2022regular}.

The RMC framework assumes that: the initial configurations, the transition relation, and the set of bad states are all described by regular languages. 
Only limited work has explored the use of non-regular specifications~\cite{FismanP01,FismanKL08}. 
Moreover, there is often an implicit assumption that starting from a regular initial configuration and repeatedly applying a regular transition relation necessarily yield a regular language. 
As the following simple \emph{two-captains protocol} illustrates, this assumption does not always hold.

Imagine a group of children forming two opposing teams during recess.
At first, there are only two captains: one for \emph{Team-A} (abbreviated as \textcolor{blue}{\texttt{a}}) on the left and one for \emph{Team-B} (\textcolor{red}{\texttt{b}}) on the right:
\scalebox{0.25}{
        \centering
    \begin{tikzpicture}
    
    \draw[step=1cm, black, thin] (3, 0) grid (5, 1);

    \fill[blue!40] (3, 0) rectangle ++(1, 1);
    \fill[red!40] (4, 0) rectangle ++(1, 1);
    \fill[pattern=north east lines, pattern color=blue] (3, 0) rectangle ++(1, 1);
    \fill[pattern=north west lines, pattern color=red] (4, 0) rectangle ++(1, 1);
     
     \draw[step=1cm, black, thick]  (3, 0) grid (5, 1);
   \end{tikzpicture}
}.
In each round, both captains simultaneously recruit a new player to join their side.
The captain of Team-A extends the line to the left by adding an \textcolor{blue}{\texttt{a}}, while the captain of Team-B  extends the line to the right by adding a \textcolor{red}{\texttt{b}}. That is: 
\scalebox{0.25}{
        \centering
    \begin{tikzpicture}
    
    \draw[step=1cm, black, thin]  (2, 0) grid (6, 1);

    \fill[blue!40] (2, 0) rectangle ++(2, 1);
    \fill[red!40] (4, 0) rectangle ++(2, 1);
    \fill[pattern=north east lines, pattern color=blue] (3, 0) rectangle ++(1, 1);
    \fill[pattern=north west lines, pattern color=red] (4, 0) rectangle ++(1, 1);
    \draw[black, ultra thick, ->] (3.5,0.5) -- (2.5,0.5);
    \draw[black, ultra thick, ->] (4.5,0.5) -- (5.5,0.5);
     
     \draw[step=1cm, black, thick] (2, 0) grid (6, 1);
   \end{tikzpicture}
}.
Which in turn becomes: \scalebox{0.25}{
        \centering
    \begin{tikzpicture}
    
    \draw[step=1cm, black, thin] (1, 0) grid (7, 1);

    \fill[blue!40] (1, 0) rectangle ++(3, 1);
    \fill[red!40] (4, 0) rectangle ++(3, 1);
    \fill[pattern=north east lines, pattern color=blue] (3, 0) rectangle ++(1, 1);
    \fill[pattern=north west lines, pattern color=red] (4, 0) rectangle ++(1, 1);
    \draw[black, ultra thick, ->] (2.5,0.5) -- (1.5,0.5);
    \draw[black, ultra thick, ->] (5.5,0.5) -- (6.5,0.5);
     
     \draw[step=1cm, black, thick] (1, 0) grid (7, 1);
   \end{tikzpicture}
}.
As rounds proceed, the two teams grow in perfect synchrony.
\begin{center}
\vspace{-1.25mm}
    \textcolor{blue}{\texttt{a}}\textcolor{red}{\texttt{b}} $\mapsto$ \textcolor{blue}{\texttt{aa}}\textcolor{red}{\texttt{bb}} $\mapsto$ \textcolor{blue}{\texttt{aaa}}\textcolor{red}{\texttt{bbb}} $\mapsto \  \cdots \ \mapsto \, \smash{\underbrace{\textcolor{blue}{\texttt{a}\cdots \texttt{a}}}_{\raisebox{0.45em}{$\scriptstyle n$}}}\;\smash{\underbrace{\textcolor{red}{\texttt{b} \cdots \texttt{b}}}_{\raisebox{0.6em}{$\scriptstyle n$}}} \, \mapsto \ \cdots$
\end{center}

\begin{wrapfigure}{r}{0.315\textwidth}
\vspace{-6mm}
    \begin{tikzpicture}[baseline={(current bounding box.center)}, scale=0.3225, every node/.style={scale=0.3225}]
     \draw[step=1cm, gray, thin] (6, 1) grid (14, -3);
        \foreach \i in {1,2,3,4} {
            \node at (5.5, 1.5-\i) {\huge ${\ldots}$};
            \node at (14.5,  1.5-\i) {\huge ${\ldots}$};
        }
        \fill[blue!40] (10, 0) rectangle ++(-1, 1);
        \fill[red!40] (11, 0) rectangle ++(-1, 1);
        \fill[blue!40] (10, -1) rectangle ++(-2, 1);
        \fill[red!40] (10, -1) rectangle ++(2, 1);
        \fill[blue!40] (10, -2) rectangle ++(-3, 1);
        \fill[red!40] (10, -2) rectangle ++(3, 1);
        \fill[blue!40] (10, -3) rectangle ++(-4, 1);
        \fill[red!40] (10, -3) rectangle ++(4, 1);
       
       \draw[step=1cm, black,semithick] (6, 1) grid (14, -3);
       \draw[step=1cm, black, thick, ->] (4.5,1) -> (4.5,-3); 
       \whiteArrowback{9.5}{-0.5}{8.5}{-0.5}
        \draw[blue, very thick, ->] (9.5,-0.5) -- (8.5,-0.5);
        \whiteArrowback{10.5}{-0.5}{11.5}{-0.5}
        \draw[red, very thick, ->] (10.5,-0.5) -- (11.5,-0.5);
        \whiteArrowback{9.5}{-1.5}{7.5}{-1.5}
        \draw[blue, very thick, ->] (9.5,-1.5) -- (7.5,-1.5);
        \whiteArrowback{10.5}{-1.5}{12.5}{-1.5}
        \draw[red, very thick, ->] (10.5,-1.5) -- (12.5,-1.5);
        \whiteArrowback{9.5}{-2.5}{6.5}{-2.5}
        \draw[blue, very thick, ->] (9.5,-2.5) -- (6.5,-2.5);
        \whiteArrowback{10.5}{-2.5}{13.5}{-2.5}
        \draw[red, very thick, ->] (10.5,-2.5) -- (13.5,-2.5);
        \node[rotate=90] at (3.85, -1) {\Huge {Time}};
        \end{tikzpicture}
        \vspace{-3mm}
\end{wrapfigure}

\vspace{1.25mm}

\noindent
After $m$ rounds, the playground is exactly \textcolor{blue}{\texttt{a}$^m$}\textcolor{red}{\texttt{b}$^m$}.

This simple example illustrates how a trivial initial configuration (of two captains) evolves, under a uniform growth rule, into the classical non-regular language $\{ a^n b^n \,|\, n\,{\in}\,\mathbb{N}\}$.

In this paper, we study the expressiveness and limitations of systems that start from a regular set of initial configurations and evolve according to a deterministic transition relation, determined by the states of processes within a fixed-radius neighborhood $r$. 
Such a transition can be viewed as a function from 
$\SigmaWindow$ to $\Sigma$, and can naturally be modeled by a regular transducer.
 We consider our evolution of the processes over a bi-infinite tape, that is, a sequence indexed by the integers $\mathbb{Z}$, rather than a one-sided infinite tape indexed by the naturals $\mathbb{N}$. 
 This bi-infinite representation is often preferable, as it facilities reasoning about systems whose evolution extends in both directions.
 The two-captains protocol already illustrates why this perspective is natural.

This bi-infinite view, together with the local update rule (each process updating its state based on its own state and the states of its neighbors within radius $r$), aligns closely with the classical framework of one-dimensional cellular automata (CA). 
The correspondence is natural; as in both settings, a collection of homogeneous processes (or cells, in CA terminology) evolves synchronously under a uniform local transition rule.

Cellular automata (CA) have long served as canonical models of distributed computation and emergent behavior~\cite{neumann1966theory,Wolfram83,BANDINI2001539,Kutrib2018}.
First, they were popularized by Conway~\cite{conway1970game} and later systemically explored by Wolfram~\cite{wolfram1984computation}.
CA have become emblematic of how simple local rules can give rise to rich, complex, and sometimes unpredictable global behavior.
In this tradition, CA are studied not only as mathematical abstractions of parallel systems, but also as laboratories for exploring the emergence of complexity from homogeneous local behaviors.
Their appeal stems precisely from this dual role as both minimalist models of computation and vivid illustrations of collective phenomena.

\looseness=-1
In parallel with these explorations, a line of work connected CA directly to formal language theory.
Smith's seminal work~\cite{Alvy1970} established a correspondence between one-dimensional CA and linearly bounded automata, situating CA within the Chomsky hierarchy and clarifying their recognition capabilities.
Subsequent contributions by Hurd \cite{Hurd1987FormalLC} and Nordahl~\cite{Nordahl1989FormalLA}, refined this perspective, showing that a one-dimensional CA defined over finite grids \emph{preserves language complexity under evolution}.
These studies demonstrated that CA, when viewed as recognizers, maintain well-defined connections to classical automata-theoretic complexity.
However, as in much of the CA literature, they treated the grid as bounded and focused on acceptance, stabilization, or halting conditions. 
The potential of CA as \emph{language generators} on an unbounded grid and an unbounded time evolution has remained comparatively underexplored.
Specifically, the evolution of a CA from a fixed initial configuration over time naturally yields sequences of global states, which may be viewed as the language of reachable states.

In this work, we investigate the \emph{expressive power} of the language generated by a CA starting from regular initial configurations. 
Specifically, we ask which families of words are \emph{generated} by a CA as it evolves on an unbounded grid? 
We refer to this perspective as the \emph{generative view} of CA languages.
This shifts the role of CA from the common focus on \emph{simulating} behaviors or \emph{recognizing languages}~\cite{KARI20053,KutribWorsch2020,SutnerKlaus2009Mcoc} to \emph{generating} languages.
A viewpoint that, while natural from a formal-methods perspective, has remained comparatively underexplored in the CA literature.
To avoid misunderstanding, we emphasize that prior work on CA, although seldom using the term \emph{generation}, has been considered under specific constraints, such as  
bounded evolutions (in time or space)~\cite{BygJorgensen2008,Gershenson2010,LuWang2011}, and analysis of \term{limit set}, which requires infinite recurrence~\cite{wolfram1984computation,mazurkiewicz1977traces,KutribMartin2021SGbC,guillon2008cellular}.
These contributions highlight important aspects of reachability in restricted settings.
In contrast, our framework adopts a \emph{symbolic} view, admits arbitrary alphabet and radii, no evolution bounds are imposed, and recurrence is not required. Thus, this allows us to position CA as genuine \emph{language generator}.
At the same time, the generative view offers a symbolic framework for analyzing the global dynamics of such systems and opens the door to future applications of formal verification in this domain.

The \emph{two-captains protocol} also sets the stage for the central notion of this paper --- \emph{glider mechanisms}. 
In the classical literature on CA, the term glider refers to a recurring multi-cell pattern that travels the grid, most famously in Conway's Game of Life~\cite{conway1970game} and later in Cook's universality construction~\cite{cook2004universality}.
Such gliders typically cycle through a sequence of shapes while advancing in the grid, returning to their initial form after several steps, a kind of ``wheel of life'' whose periodicity underlies their evolution.

Our usage is inspired by these traditions, but adopts a more structured perspective. 
Here, a glider is not a composite pattern but a \emph{single cell}, endowed with the semantics of the automaton's local rules: its symbol, its velocity, i.e., speed and direction (left or right), and its interaction behavior.
In this sense, the cycle sustaining a Conway-style glider is compressed to size one: the state of the glider remains fixed as it propagates. 
On the other hand, since our automata may have arbitrary radius, a glider's velocity need not be limited to one cell per step, but may reflect any displacement determined by the local rule.

As in the classical setting, our glider persists indefinitely in the absence of interaction, but may be terminated when colliding with other gliders. 
The key difference is that persistence, velocity, and interaction behavior are all dictated directly by the CA local rule function, rather than being emergent from multi-cell dynamics.
This makes gliders in our setting \emph{symbolic building blocks}; their semantics are modular and explicit, {even when the underlying CA has a large interaction radii}. 
Seen through this lens, the captains of the \emph{two-captain protocol} are instances of gliders: we have two gliders of state \textcolor{blue}{\texttt{a}} one of velocity of $-1$ and the second of velocity $0$, that way, in every iteration, Team-A ``expands'' one step to the left (the $-1$ velocity) while maintaining all existing locations (the $0$ velocity). 
Similarly, for state \textcolor{red}{\texttt{b}} we have two gliders as well, one of velocity $0$ and the second of velocity $+1$, which allows the expansion of Team-B towards the right.  
These structured gliders provide the building blocks for the generative semantics developed in the reminder of the paper, where we show how these gliders yield a modular explanation of emergent pattern and enable constructive proofs of expressive power in a system where population growth is essential.
\vspace{-1.5mm}

\paragraph{Expressiveness Results.}
In exploring the generative power of CA, we set out to find languages that push the expressive boundary. We began with a simple non-regular example, showing that $a^n b^n$ arises naturally for local growth rules. 
From there, we traveled alongside variations of the one-counter automata languages, a class that strictly contains the regular languages and reaches into the context-free. 
Here, we construct CA generating $a^n b a^n$, $a^n b^m$ for $n\geq m$, demonstrating how these patterns can be captured in one-dimensional growth. Pushing further, we exhibited even context-sensitive behavior; as language $a^n b^n c^n$ can be generated in this framework. 
More generally, we show that languages of the form \begin{center}  
$L = \{w_1^{e_1(n)}{\ldots} w_m^{e_m(n)} \mid n\,{\in}\,\mathbb{N}\}$
\end{center}
are CA-expressible (and, in particular, glider-expressible), for any $m\in\mathbb{N}$,  any words $w_1\,{\ldots}\, w_m$ , and any positive linear expressions $e_i(n)$. 
We further note that glider-expressible languages form a strict subclass of CA-expressible languages.

\vspace{-0.75mm}

\paragraph{Contributions.}Beyond these expressiveness results, the paper makes the following two contributions.
\begin{itemize}[topsep=1.5mm,itemsep=1.5mm]
        \item  \emph{Generative framework.} We introduce a formulation of cellular automata as genuine \emph{language generators}. 
    \item \emph{Glider mechanisms.} We define gliders as atomic carriers specified by a symbol, velocity, and interaction semantics derived from the CA rule. 
\end{itemize}
    Combining the two, we show how gliders provides modular explanations for CA-generated languages and enable constructive proofs of expressive power.

\looseness=-1
To the best of our knowledge, this is the first formal analysis of CA-expressible languages over unbounded, bi-infinite grids with arbitrary finite-size alphabets and radii. In contrast to earlier work on bounded grids or finite-time evolutions \cite{Hurd1987FormalLC,Nordahl1989FormalLA}, which proved that language complexity is preserved, we show that even regular initializations (which sit at the bottom of the Chomsky hierarchy) can lead to languages of strictly higher complexity. This work aims to highlight an underexplored connection between local dynamics and global formal properties.

\vspace{-0.75mm}

\paragraph{Paper Structure.} 
The remainder of the paper is organized as follows. \Cref{sec:prelim} provides preliminaries. \Cref{sec:expressivness} introduces the cellular automata (CA) framework and its generative point of view. \Cref{subsec:onecounter} provides a walkthrough one counter automata languages, which leads us to \Cref{subsec:gliders} where we define gliders mechanism and establishes the strict separation between CA-expressible and glider-expressible. \Cref{subsec:expressivness} develops our expressiveness results from illustrative cases to general families. Then we conclude our discussion at \Cref{sec:discussion}.

\section{Preliminaries}\label{sec:prelim}
We begin by fixing key notational conventions and reviewing background from formal language theory. We then formally define the cellular automaton (CA) model considered in this work. These conventions are provided to avoid ambiguity in the technical development that follows.

\vspace{-1mm}

\paragraph{Conventions.}
Let $\mathbb{N}$ be the set of positive natural numbers, and let $\mathbb{N}_0\,{=}\,\mathbb{N}\,{\cup}\, \{0\}$ be the set of non-negative integers.
Let $\mathbb{Z}$ denote the set of all integers, with $\mathbb{Z}^{+}{=}\,\mathbb{N}$ and $\mathbb{Z}^{-}$ representing the positive and negative integers, resp.
For any $n,m\,{\in}\,\mathbb{Z}$ where $n\,{\leq}\,m$, we write $[n,m]$ to denote the interval $\{n,n{+}1,{...}, m\}$, and we abbreviate $[1,m]$ as $[m]$.
Given a domain $D \,{=}\, [x, y] \subseteq \mathbb{Z}$ we use $D{+s}$ for the \term{domain shift} of $D$ by $s\in\mathbb{Z}$, i.e., for the interval $[x{+s}, y{+s}]$.
When convenient, we use the notation $(-\infty, k]$ and $[k,\infty)$ for $k\in \mathbb{Z}$ in the classical manner, e.g., $\mathbb{N}$ can be seen as $[1,\infty)$.

\paragraph{Words.}
An \term{alphabet} $\Sigma$ is a finite, non-empty set of elements called \emph{letters}.
A finite \term{\wword} $w$ over $\Sigma$ is an assignment of $\Sigma$ to consecutive locations. It could be considered as a function $w \,{:}\, [n] \to \Sigma$ for some $n\in \mathbb{N}_0$.
That is, $w\,{=}\,\sig{1}{\ldots} \sig{n}$ where ${\sig{i}\,{=}\,w[i]}$ is from $\Sigma$. Its \term{length} is denoted by $|w|\,{=}\,n$.
{The empty word is denoted by $\varepsilon$ and $|\varepsilon|\,{=}\,0$. }
The set of all finite words over $\Sigma$ is $\Sigma^*$, and $\Sigma^+\,{=}\,\Sigma^*\setminus\{\varepsilon\}$.
  
\term{Infinite words} are defined over infinite domains. Words over domain $\mathbb{Z}$ are referred to as 
\emph{bi-$\omega$ words} (or \term{bi-infinite}). When the domain is $(-\infty,k]$ (resp. $[k,\infty)$) for some $k\in\mathbb{Z}$ the words are referred to as \term{left-$\omega$} (resp. \term{right-$\omega$}) words.

\vspace{-1mm}

\paragraph{Infix.} 
Let $\mathbb{D},\mathbb{D}'\subseteq \mathbb{Z}$ be two integer intervals.
Let ${w\,{:}\,\mathbb{D}\,{\to}\, \Sigma}$ and ${w'\,{:}\,\mathbb{D}'\,{\to}\, \Sigma}$ be two (finite or infinite) words. 
We say that $w'$ is an \term{infix} of $w$ if {there exists $s\,{\in}\,\mathbb{Z}$ such that} $\mathbb{D}'\, {+s} \,{\subseteq} \, \mathbb{D}$ and $ {w'[i{+s}]\,{=}\,w[i]}$ for every $i\,{\in}\,\mathbb{D}'$.
We write $w' \, {\sqsubset}\, w$ to denote that $w'$ is an \emph{infix} of $w$.

\vspace{-1mm}

\paragraph{Concatenation.} 
    Given two finite words $x,y\,{\in}\,\Sigma^*$, their \term{concatenation} is denoted $x\,{\cdot}\,y$ (or simply $xy$).
    For $k \, {\in} \, \mathbb{N}$, we write $x^k$ for the $k$-times right concatenation of $x$ with itself.
    Similarly, ${x^{\omega}}$, ${^{\omega} x}$, and ${{^{\omega}}x^{\omega}}$ for an infinite concatenation of $x$ to itself to the right, left, or both sides, respectively.

    The concatenation of a finite or left-$\omega$ (resp. right-$\omega$) word to the left (resp. right) of a right-$\omega$ (resp. left-$\omega$) word is defined naturally.

\vspace{-1mm}

\paragraph{Languages.}
A \term{\macrolang} over $ \Sigma$, is a set of words over $\Sigma$.
We use ${\Sigma^*}$, ${\Sigma^{\omega}}$,${{\,}^{\omega}\Sigma}$~and ${^{\omega}\Sigma^{\omega}}$ to denote the sets of all finite, right-${\omega}$, left-${\omega}$ and bi-${\omega}$ words, resp.
Subsets of ${\Sigma^*}$, are \term{\fLang}. The
subsets of ${\Sigma^{\omega}}$ and ${{^{\omega}}\Sigma}$ are referred to as \term{\infLangs}, consisting of right- and left-${\omega}$ words, resp. 
A subset of ${^{\omega}\Sigma^{\omega}}$ is a \term{\biInfLang}. 
For any ${W\,{\subseteq}\,{\Sigma^*}}$ and ${k\,{\in}\,\mathbb{N}}$, we write $W^{k}\,{=}\allowbreak\{ w_1 {\cdot}w_2 {\cdots} w_k\,|\,\forall i\,{\in}\,{\lbrack k \rbrack}.\ {w_{i}}\,{\in}\, W\}$ for the set of all $k$ concatenations of words from $W$. 
The sets ${W^{\omega}}$, ${^{\omega} W}$, and ${^{\omega} W^{\omega}}$ are defined analogously for infinite repetition.

\subsection{Cellular Automata (CA)}
\label{CA:subsec} 
In the following we formally introduce and define our models of interest,  \emph{cellular automata}. 
These are dynamical systems, exhibiting a variety of organized and complex behaviors over time~\cite{Hedlund1969,kuurka1997languages,alma9926314815204361,Kurka1999,kurka2003topological,BealMarie-Pierre2013SDaF}.
We restrict our attention to one-dimensional cellular automata {({1}-{CA})} defined over bi-infinite grids, with arbitrary finite alphabets and neighborhood radii. We henceforth refer to them simply as CA, omitting the dimensional qualifier.

\begin{definition}[Cellular Automata (CA)]
    A CA $\mathcal{A} = (\Sigma ,r ,f)$
    consists of an alphabet $\Sigma$ of \term{states}, a \term{radius} $r\,{\in}\, \mathbb{N}_0$, and a \term{local rule} $f\!:\!{\SigmaWindow}{\to}\,\Sigma\!$. 
\end{definition}

Let $\mathcal{A} = (\Sigma, r, f)$ be a CA, we use $\mathcal{A}$ to define the following definitions.

\paragraph{Configurations.}
    A \term{configuration} $c$ of {$\mathcal{A}$,}
is an assignment of $\Sigma$ elements to each cell indexed by $\mathbb{Z}$. 
    A configuration is viewed as a bi-$\omega$ word, thus $c\in \,^{\omega}\Sigma^{\omega}$.
    Throughout this paper, we denote this \emph{set of bi-infinite words by $\SigmaZWindow$}, which is the standard notation in the CA literature. 

    \smallskip
Given a configuration ${c\,{\in}\, {\SigmaZWindow}}$ and $i,j  \, {\in} \, {\mathbb{Z}}$ with ${i \,{<}\, j}$, we write ${c[i]}$ (or ${\langle i \rangle_{\scriptstyle c}}$) to denote the cell state at position $i$.
Similarly, $c{[i,j]}$ or ${\langle i,j \rangle_{\scriptstyle c}}$ denotes the infix of the configuration from index $i$ to $j$, inclusive. Hence, ${\langle i,j \rangle_{\scriptstyle c} \,{\in}\, {\Sigma^{j{-}i{+}1}}}$.

\paragraph{Neighborhood.}
Each cell $i\in\mathbb{Z}$ of $\mathcal{A}$, 
updates its state according to the local rule function (also referred to as the \emph{rule function}),
which acts according to the current states of a fixed set of adjacent cells centered in cell $i$; this is the \term{neighborhood} of $i$. 
For a \term{radius} $r$, the state update depends on the ${2r{+}1}$ consecutive cells centered at the current position. Formally, given a configuration $c$ the \term{neighborhood} {of cell $i$} is defined as {$\langle i-r, i+r\rangle_c$}.

\paragraph{Quiescent States and Configuration Support.}
Given {$\mathcal{A}$}, a state $s\,{\in}\, \Sigma$ is termed a \term{quiescent state} if ${f({s^{2r{+}1}})\,{=}\,s}$, that is, the rule function $f$ maps a neighborhood of $s$'s to state $s$. We follow the common assumption that there is a unique quiescent state, and denote it by $\mybot$.
Given a configuration $c$, the \term{support} of $c$ is the set of cells not in state $\mybot$, denoted 
${\supp_{\bot}(c) \,{=}\,  \{i \, {\in} \, \mathbb{Z}  \,|\,  c[i] \,{\neq}\,\mybot \}}$.

\paragraph{Finite configurations} are those with only finitely many non-quiescent states. Formally, a configuration $c$ is \term{finite} if and only if its support; $\supp_{\bot }(c)$, is a finite set.
We denote the set of all finite configurations by $\mathterm{C_{\mathcal{F}}}$. Note that $C_{\mathcal{F}} \,{\subseteq}\, \SigmaZWindow$. 
Let $c\,{\in}\,C_{\mathcal{F}}$, we define the \term{active interval} of $c$ as the smallest integer interval containing its support, formally defined as:  
\begin{equation}
    \label{eq:interval} 
   \interval{c} \coloneqq [\min (\supp_{\bot}(c)), \max (\supp_{\bot}(c))]
\end{equation}
Let $c,c'\,{\in}\,C_{\mathcal{F}}$ be two finite configurations, and let $\interval{c},\interval{c'}$ be their intervals. If there exists $s \,{\in}\, \mathbb{Z}$ such that $\interval{c}{+}s \,{=}\, \interval{c'}$, and $c[\interval{c}+s] \,{=}\, c' [\interval{c'}]$, then the two configurations are considered \emph{equivalent} under a domain shift.

\noindent
Consequently, \term{width} of $c$, denoted $\width{c}$, is defined as: 
\begin{equation}
\label{eq:width} 
   \width{c} \coloneqq \max (\supp_{\bot}(c))  - \min (\supp_{\bot}(c)) + 1
\end{equation}
\paragraph{Orbits.} The local rule function $f$ of $\mathcal{A}$ is extended to the \term{global rule function}, denoted ${\mathterm{G}\,{:}\,\SigmaZWindow\,{\rightarrow}\,\SigmaZWindow}$, where $f$ is applied to all cells simultaneously. 
Starting from an initial configuration ${c_0 \!\in\! {\SigmaZWindow}}$, applying the global function once, we get the next configuration, that is, $c_1 \,{=}\,  G(c_0)$ and generally, $c_n\,{=}\,G(c_{n{-}1})\,{=}\,{G^{n}}(c_0)$ for every $n \!\in\! \mathbb{N}$. The
\term{orbit} of a configuration $c$, denoted ${\orbitc{c}{=} \{{G^n} (c)  ~|~ n \!\in\! {\mathbb{N}_{0}}\}}$, is the infinite configuration sequence $c_0,c_1,c_2,\ldots$ where $c_0\,{=}\,c$ and ${c_n \,{=}\, {G^n}(c)}$ for every $n \!\in\! \mathbb{N}$.
Given a set of configurations $C\,{\subseteq}\, {\SigmaZWindow}$, its orbit is denoted ${\mathterm{\orbitc{C}}{=}\bigcup_{c \in  C} \orbitc{c}}$.
Note that an orbit of any ${c \!\in\! {\SigmaZWindow}}$ is a bi-infinite language.

\section{CA Languages -- Generative Perspective}\label{sec:expressivness}
We now formalize what it means for a CA to \emph{generate a language}, viewing its reachable configurations as a formal object of study. This framework allows us to examine how parameters such as alphabet size and neighborhood radius influence expressiveness and to relate the generated languages to established classes in formal language theory.
\paragraph{Padding and Finite Configurations.}First we define a $\bot$-padding operator on words, then we lift it to languages.
The function, $\pad : \Sigma^* \to C_{\mathcal{F}}$, takes a finite word $w\in\Sigma^*$ and returns ${}^{\omega} \bot w \bot^{\omega}$.
The lifting of the function to languages is as expected $\ppad : \mathcal{P}(\Sigma^*) \to \mathcal{P}(C_{\mathcal{F}})$ where $\ppad(L) := \{\pad(w) \mid w\,{\in}\, L\}$. 

Given a \fLang\ $F$, we define the set of finite configurations induce by $F$ to be $\ppad(F)$, as explained above. 
In the sequel, we consider CA languages derived from initial configurations ${{I} \,{=}\, \mathterm{\ppad}(F)}$ for a \emph{regular language $F$}. Regular languages are in particular finitary languages.
\paragraph{CA language.}
Given ${\mathcal{A}\,{=}\,  (\Sigma,r,f)}$ and ${I\,{\subseteq}\, C_{\mathcal{F}}}$.
The language of $\mathcal{A}$ is defined wrt the set of initial configurations $I$ and is denoted  $\mathterm{\mathcal{L}(\mathcal{A},I)} \,{=} \,\allowbreak\{ w \mid \, {}^{\omega} \mybot  w \mybot ^{\omega} \, {\in}\,  \orbitc{I}\}$.
Given ${F \,{\subseteq}\, \Sigma^*}$, we abbreviate ${\mathcal{L}(\mathcal{A}, \ppad(F))}$ 
as
${\mathcal{L}(\mathcal{A}, F)}$.

\begin{remark} \normalfont
Although our definitions allow general initial sets ${I \,{\subseteq}\, C_{\mathcal{F}}}$, we focus on the structured case $ {I \,{=}\, \ppad(F) }$, for a regular language $ {F \,{\subseteq}\, \Sigma^*}$. 
Regular initial sets are standard in formal verification, as many natural systems admit such representations~\cite{Bouajjani2000}.
Despite this syntactic constraint, we show that such cases can still yield CA-expressible languages of considerable complexity.
\end{remark}

\paragraph{$\mathbb{CA}$ and CA-expressible.} We use $\caclassshort{d}{r}{\Sigma}$ to denote the class of 1-CA with radius $\mathterm{r}$ over $\Sigma$, respectively. ${\mathbb{CA}}$ denotes the class of all 1-CA (for any $r$ and $\Sigma$). 
A language $L \,{\subseteq}\, \Sigma^*$ is $\caclassshort{1}{r}{\Sigma}$-expressible (thus, CA-expressible in general), if there exists $\mathcal{A} \!\in\! \caclassshort{1}{r}{\Sigma}$ and $I \,{\subseteq}\, C_{\mathcal{F}}$ such that $L \,{=}\, \mathcal{L}(\mathcal{A},I)$. In which case we write  $L \!\in\! \caclassshort{1}{r}{\Sigma}$.
When ${I\,{=}\,\ppad(F)}$ and $F$ is a regular language, we often write ${\mathcal{L}(\mathcal{A},I)\,{\in}\,\caclassshortREG{1}{r}{\Sigma}{REG}}$ to emphasize the regular origin.
\subsection{Feasible Neighborhoods}
Let ${\mathcal{A}\,{\in}\,\caclassshort{1}{r}{\Sigma}}$ be a fixed CA. 
By the above definitions,  not every neighborhood ${w\,{\in}\, \SigmaWindow}$, where $w\,{=}\,{\sig{{-}r}{\ldots}\sig{0}\sig{1} {\ldots}\sig{r}}$, needs to appear along orbits starting from an item of $I$.
We say that a neighborhood $w \!\in\! \Sigma^{2r+1}$ is \term{feasible} for $\mathcal A$ under $I$ if there exists $c \!\in\! \orbitc{I}$ satisfying $w \sqsubset c$; otherwise $w$ is \term{infeasible}.
Accordingly, we define the set of \term{feasible neighborhood} of ${\mathcal{A}\,|_{I}}$ (read $\mathcal A$ under $I$).
\begin{definition}[Feasible Neighborhoods]
    For $\mathcal A \!\in\! \caone{r}{\Sigma}$ and $I \!\subseteq\! C_{\mathcal F}$, the set of feasible neighborhoods is
    $\FN_{\mathcal A,I} \!=\! \{\,w \!\in\! \Sigma^{2r+1}\mid \exists\,c \!\in\! \orbitc{I}\text{ with }w\sqsubset c\,\}.$
\end{definition}

\noindent
\begin{minipage}{0.725\textwidth}
    \begin{example}\label{example:feasible- neighborhood} \normalfont
        Let $\mathcal{A}\,{\in}\,\caclassshort{1}{1}{\{a,\bot\}}$, 
        where $f$ is defined by:
          
          \vspace{2mm}
          
          \scalebox{0.925}{
      $
       \begin{array}{llll}
        f(\bot \bot \bot )  \,{=}\, \bot \ &\  f(\bot \bot a)  \,{=}\, a \ &\  f(\bot a\bot )  \,{=}\, \bot \ &\  f(\bot aa)  \,{=}\, a\\[0.65 mm]
        \text{\makebox[\widthof{$f(\bot \bot \bot )$}][c]{$f(a\bot \bot )$}}\,{=}\, \bot  \ &\  \text{\makebox[\widthof{$f(\bot \bot a)$}][c]{$f(a\bot a)$}}  \,{=}\, a \ &\  \text{\makebox[\widthof{$f(\bot a \bot)$}][c]{$f(aa\bot )$}}  \,{=}\, \text{\makebox[\widthof{$\bot$}][c]{$a$}} \ &\  \text{\makebox[\widthof{$f(\bot aa)$}][c]{$f(aaa)$}}  \,{=}\, a
        \end{array}
        $}

        \vspace{2mm}
    \end{example}
\end{minipage}
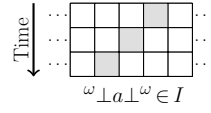
\begin{minipage}{0.25\linewidth}
    \centering
    \begin{tikzpicture}[scale=0.325, every node/.style={scale=0.325}]
     \draw[step=1cm, gray, thin] (7, 1) grid (12, -2);
        \foreach \i in {1,2,3} {
            \fill[lightgray!50] (11 -\i, -\i+1) rectangle ++(1, 1);
            \node at (6.5, 1.5-\i) {\huge ${\ldots}$};
            \node at (12.5,  1.5-\i) {\huge ${\ldots}$};
        }
       \draw[step=1cm, black, thin] (7, 1) grid (12, -2);
        \node at (9.45, -2.75) {\Huge \textbf{ $^{\omega}\bot a\bot^{\omega}\,{\in}\, I$}};
        \draw[step=1cm, black, thick, ->] (5.5,1) -> (5.5,-2); 
        \node[rotate=90] at (5.0, -0.5) {\Huge {Time}};
        \end{tikzpicture}
\end{minipage}

    \noindent
    Let $I \,{=}\, \ppad(\{a\})$. Thus, only $\bot \bot \bot$, $\bot \bot a$, $\bot a\bot$, and $a\bot \bot$ are feasible neighborhoods (\FNs) for $\mathcal{A}|_I$. Consider the illustration to the right, of applying the rule function on $I$, where $\whitebox \,{=}\, \bot$ and $\abox \,{=}\, a$.
    Therefore, it suffices to define $f$ as a partial function, on the \FN, namely: 
    \begin{center}
    \vspace{-1.5mm}
        \scalebox{0.905}{
      $
       \begin{array}{llll}
        f(\bot \bot \bot )  \,{=}\, \bot  \ &\  f(\bot \bot a)  \,{=}\, a \ &\  f(\bot a\bot )  \,{=}\, \bot  \ &\  f(a\bot \bot )  \,{=}\, \bot
        \end{array}$}
        \vspace{-0.5mm}
    \end{center}

\noindent
Infeasible neighborhoods (those outside $\FN_{\mathcal A,I}$) can be ignored in the analysis. Therefore, henceforth we define $f$ as a partial function with domain $\FN_{\mathcal A,I}$.

  Given a $\mathcal{A}\,{\in}\, \caclassshort{1}{r}{\Sigma}$, the local rule $f$ assigns to each neighborhood $w\,{\in}\, \SigmaWindow$ a value from $\Sigma$. It is often more convenient to ask for each $\sigma\,{\in}\,\Sigma$ which (feasible)

  \noindent
\begin{wrapfigure}{r}{0.275\linewidth}
      \centering
      \vspace{-5mm}
            {\small {$\begin{array}{|l|l|}
       \hline
        \bot & \bot \bot \bot  \quad \bot a\bot \quad a\bot \bot  \\
        a&  \bot \bot  a\\
        \hline
        \end{array} $}}

        \vspace{1mm}
        
        \textbf{\footnotesize $f$ from Ex.\ref{example:feasible- neighborhood}}
         
        \vspace{-4mm}
\end{wrapfigure}
     neighborhoods are assigned to it. 
     Consider Ex.\ref{example:feasible- neighborhood}, its local rule can be succinctly given by a table consisting of only feasible neighborhoods.
     Each row corresponds to a letter $\sigma\,{\in}\,\Sigma$ and lists all the {\FNs} for which ${f(w)} \,{=}\, \sigma$.

   \subsection{Observations and Implications}\label{subsec:CAprop}
    Next, we present a few elementary observations about the CA and its generated language. Although straightforward, they underpin key intuitions that support the more involved developments that follow. 
    Full proofs are provided in App.\ref{app:CAproperties}.

    Given a language $L\,{\in}\,\Sigma^*$, let $w_0,w_1,w_2,\ldots$ be a list of all words in $L$, arranged in a non-decreasing order by their length, (i.e., $|{w_{i}}|\,{\leq}\,|w_{i{+}1}|$ for any $i \, {\in} \, \mathbb{N}_0$). 
    Let ${\apos{1}}, {\apos{2}},\ldots$ be the sequence of size differences of words in this sequence, i.e., ${\apos{i}} \,{=}\, |w_{i}| {-}|w_{i \,{-}\, 1}|$ which we term the \term{differences sequence}.
    We say that the difference sequence of $L$ is \emph{bounded} if there is a $k \, {\in} \, \mathbb{N}$ such that $ {\apos{i}}\,{\leq}\,k$ for all $i \, {\in} \, \mathbb{N}$.
    If the difference sequence is bounded, we define the minimal such $k$ as the \term{difference-bound} of $L$. 
    We proceed with some observations regarding $\interval{c}$ and $\width{c}$ defined in \Cref{eq:interval,eq:width}.
    \begin{note}
    \label{cor:bound_inc}
     Let $\mathcal{A} \,{\in}\, \caclassshort{1}{r}{\Sigma}$ and let $c\,{\in}\, C_{\mathcal{F}}$ be a finite configuration. 
    Assume \emph{$\interval{c}\,{=}\,[m,n]$} and \emph{$\interval{G(c)}\,{=}[i,j]$}. 
     Then $i \,{\geq}\,m {-}r$ and $j \,{\leq}\,n{+}r$.
    \end{note}
\vspace{-1.95mm}    
\begin{restatable}
    {corollary}{bounddiff}\label{CApump:lemma}
        Let $\mathcal{A}\,{\in}\, \caclassshort{1}{r}{\Sigma}$. 
        Then there exists a natural number $k\,{\leq}\,2\,{\cdot}\, r$ such that each finite configuration $c\,{\in}\, C_{\mathcal{F}}$ of $\mathcal{A}$ satisfies \emph{$\width{G(c)} \minus \width{c} \,{\leq}\,k$}.
    \end{restatable}

\vspace{-1mm}    
   Cor.\ref{CApump:lemma} can be used to show that a language $L$ is not $\caclassshort{1}{r}{\Sigma}$-expressible for certain $r$'s.
    That is, let $L \,{=}\, \mathcal{L}(\mathcal{A},\{c\})$ for some $\mathcal{A}\,{\in}\,\caclassshort{1}{r}{\Sigma}$ and $c\,{\in}\,C_{\mathcal{F}}$. 
    If $L$ is $k$-difference bounded, 
    then it is necessarily not CA-expressible with $r\,{\leq}\,{\lfloor{\frac{k \,\minus \,1}{2}}\rfloor}$.

\section{OCA Separation Witnesses via CA}\label{subsec:onecounter}
    In this section, we examine the expressive capabilities of regular-initialized CA.   
    First note that every regular language, $L\,{\in}\,\Sigma^*$, is trivially $\caclassshort{1}{0}{\Sigma}$-expressible, by setting the initial configuration to $I\,{=}\,\ppad(L)$ and using the identity rule function: 
    $\forall \sigma\,{\in}\,\Sigma. \ f(\sigma)\,{=}\,\sigma$. That is, $\caclassshortREG{1}{r}{\Sigma}{REG} \supseteq \class{REG}$ where $\class{REG}$ stands for the class of regular languages. 
    Next, we use one-counter automata (OCA) and their known subclasses as a framework for comparison. 
    
    \paragraph{One-counter automata (OCA)} are a fundamental model of infinite-state systems. They are a special case of pushdown automata (PDA) where the stack is restricted to a single symbol.
    As a result, the stack behaves as a single unbounded natural counter, which can be incremented, decremented (only if positive), and tested for zero.
    While seemingly simple, OCA recognizes a rich class of languages that extends the regular languages and captures basic counting.
    This makes OCA a natural tool for probing the expressive scope of our model. By positioning regular-initialized CA generated languages against increasingly expressive subclasses beyond the regular class, we demonstrate a non-trivial increase in language expressiveness.

\paragraph{OCA Subclasses.}
    Although one-counter automata (OCA) are inherently non-{\allowbreak}deterministic,
    their 
    \term{deterministic} variant DOCA-a strict subclass of OCA-has been widely studied.\footnote{Their study is also motivated by the fact that many decision problems are undecidable for general OCAs and PDAs~\cite{valiant1975deterministic,mathew2025learning,Berman1988DecidingEO}.}
    Within DOCA,
    further syntactic restrictions yield \term{real}-\term{time} variant ROCA.
    Finally, restricting ROCA yields a \term{visible} model VOCA.
    Despite being the most restricted type, even VOCA strictly subsumes the regular languages.
    For clarity and conciseness, 
    we provide the full formal definitions of OCA and their relevant subclasses in App.\ref{app:oca-defs}.

\paragraph{OCA languages.} Given an OCA $\mathcal{A}$, the language it recognizes is denoted by ${\mathcal{L}(\mathcal{A})}$. The class of languages recognized by OCA is denoted by $\class{OCL}$. 
    Similarly, $\class{DOCL}$, $\class{ROCL}$, and $\class{VOCL}$ are defined based on DOCA, ROCA and VOCA. 

    \vspace{0.5mm}
    
    These subclasses form a hierarchy of increasing restrictions and decreasing expressive power, see \cite{hopcroft_ullman_automata_2001,Herbst1991,staquet2024active,VisiblyPushdownLanguages04}
 for formal definitions and proofs. 
\begin{restatable}[OCL hierarchy \cite{staquet2024active}]{theorem}{OCAhierarchy}\label{Thm:one-counter-hierarchy} 
    ${\class{REG}}\,{\subsetneq}\, {\class{VOCL}}\,{\subsetneq}\,
{\class{ROCL}}\,{\subsetneq}\, {\class{DOCL}}\,{\subsetneq}\, {\class{OCL}}$
\end{restatable}

\noindent
Thm.\ref{Thm:one-counter-hierarchy} establishes a strict hierarchy.
For each consecutive pair of classes, one can show a CA-expressible language that separates them, as we next show.

\begin{figure}[ht] 
\centering
 \captionsetup[subfloat]{font=footnotesize}
   \subfloat[$\{a^n b^n \,|\, n\,{\in}\,\mathbb{N}\}$\label{fig:OCA-lang1-a}]{
  \centering
        \scalebox{0.9025}{$\begin{array}{|l|l|}
           \hline
            \mybot  & \mybot \mybot \mybot  \\[0.3mm] \hline 
        a & \mybot \mybot a \quad \mybot aa \quad \text{\makebox[\widthof{$b\mybot \mybot$}][c]{$\mybot a b$}} \quad aaa \quad aab \\[0.3mm] \hline 
        b & \text{\makebox[\widthof{$\mybot \mybot a$}][c]{$ab\mybot$}}  \quad  \text{\makebox[\widthof{$\mybot a a$}][c]{$abb$}} \quad b\mybot \mybot  \quad  bb\mybot   \quad bbb \\
            \hline
                \end{array}$}
    \hspace{1mm}
    \scalebox{0.3525}{
     \begin{tikzpicture}[baseline={(current bounding box.center)}]
     \draw[step=1cm, gray, thin] (6, 1) grid (14, -2);
        \foreach \i in {1,2,3} {
            \node at (5.5, 1.5-\i) {\huge ${\ldots}$};
            \node at (14.5,  1.5-\i) {\huge ${\ldots}$};
        }
        \fill[lightgray!50] (10, 0) rectangle ++(-1, 1);
        \fill[darkgray!50] (11, 0) rectangle ++(-1, 1);
        \fill[lightgray!50] (10, -1) rectangle ++(-2, 1);
        \fill[darkgray!50] (10, -1) rectangle ++(2, 1);
        \fill[lightgray!50] (10, -2) rectangle ++(-3, 1);
        \fill[darkgray!50] (10, -2) rectangle ++(3, 1);
       
       \draw[step=1cm, black, thick] (6, 1) grid (14, -2);
       \draw[step=1cm, black, ultra thick, ->] (4.75,1) -> (4.75,-2); 
        \node[rotate=90] at (4.25, -0.5) {\Huge {Time}};
        \node[] at (10,-2.7) {\Huge {$\,{} ^{\omega}\bot ab\bot^{\omega} \, {\in}\, I$}};
        \end{tikzpicture}
    }}

 \subfloat[$\{a^n b a^n \,|\, n\,{\in}\,\mathbb{N}\}$\label{fig:OCA-lang1-b}]{
  \centering
        \scalebox{0.9025}{{$\begin{array}{|l|l|}
            \hline
                \mybot  & \mybot \mybot \mybot  \\\hline
                a & \mybot \mybot a  \quad  \mybot aa  \quad \mybot ab  \quad a\mybot \mybot   \quad aa\mybot\\[-0.1mm]
                & \text{\makebox[\widthof{$\mybot \mybot a$}][c]{$aaa$}}  \quad \text{\makebox[\widthof{$\mybot a a$}][c]{$aab$}}  \quad \text{\makebox[\widthof{$\mybot a b$}][c]{$ba\mybot$}}   \quad \text{\makebox[\widthof{$a\mybot\mybot$}][c]{$baa$}} \\ \hline
                b & \text{\makebox[\widthof{$\mybot\mybot a$}][c]{$aba$}} \\
                \hline  
            \end{array}$}}
            \hspace{1mm}
    \scalebox{0.3525}{
     \begin{tikzpicture}[baseline={(current bounding box.center)}]
     \draw[step=1cm, gray, thin] (6, 1) grid (15, -2);
        \foreach \i in {1,2,3} {
            \node at (5.5, 1.5-\i) {\huge ${\ldots}$};
            \node at (15.5,  1.5-\i) {\huge ${\ldots}$};
        }
        \fill[lightgray!50] (10, 0) rectangle ++(-1, 1);
        \fill[lightgray!50] (11, 0) rectangle ++(1, 1);
        \fill[darkgray!50] (11, 0) rectangle ++(-1, 1);
        \fill[lightgray!50] (10, -1) rectangle ++(-2, 1);
        \fill[lightgray!50] (11, -1) rectangle ++(2, 1);
        \fill[darkgray!50] (10, -1) rectangle ++(1, 1);
        \fill[lightgray!50] (10, -2) rectangle ++(-3, 1);
        \fill[lightgray!50] (11, -2) rectangle ++(3, 1);
        \fill[darkgray!50] (10, -2) rectangle ++(1, 1);
       
       \draw[step=1cm, black, thick] (6, 1) grid (15, -2);
       \draw[step=1cm, black, ultra thick, ->] (4.75,1) -> (4.75,-2); 
        \node[rotate=90] at (4.25, -0.5) {\Huge {Time}};
        \node [] at (10.5,-2.7) {\Huge {$\,{} ^{\omega}\bot aba\bot^{\omega}\, {\in}\, I$}};
        \end{tikzpicture}
    }}
\hspace{1mm}
\subfloat[$\{a^n b^m c \,|\, n\,{\geq}\, m \,{>}\,0\}$\label{fig:OCA-lang1-c}]{
  \centering
        \scalebox{0.9025}{$\begin{array}{|l|l|}
            \hline
            \mybot  & \mybot \mybot \mybot \\[0.3mm] \hline
            a & \mybot \mybot a \quad \mybot aa \quad \mybot ab \quad aaa \quad aab\\[0.3mm] \hline
            b & \text{\makebox[\widthof{$\mybot \mybot a$}][c]{$a b b$}} \quad  \text{\makebox[\widthof{$\mybot a a$}][c]{$a b c$}} \quad \text{\makebox[\widthof{$\mybot a b$}][c]{$bbb$}}  \quad \text{\makebox[\widthof{$aa a$}][c]{$bbc$}} \quad bc\mybot \\[0.3mm] \hline
            c & c\mybot \mybot  \\  
            \hline
                \end{array}$}
            \hspace{1mm}
    \scalebox{0.3525}{
     \begin{tikzpicture}[baseline={(current bounding box.center)}]
     \draw[step=1cm, gray, thin] (6, 1) grid (16, -2);
        \foreach \i in {1,2,3} {
            \node at (5.5, 1.5-\i) {\huge ${\ldots}$};
            \node at (16.5,  1.5-\i) {\huge ${\ldots}$};
        }
        \fill[lightgray!50] (9, 0) rectangle ++(2, 1);
        \fill[darkgray!50] (11, 0) rectangle ++(1, 1);
        \fillc{12}{0}{1}
        
        \fill[lightgray!50] (8, -1) rectangle ++(3, 1);
        \fill[darkgray!50] (11, -1) rectangle ++(2, 1);
        \fillc{13}{-1}{1}

        \fill[lightgray!50] (7, -2) rectangle ++(4, 1);
        \fill[darkgray!50] (11, -2) rectangle ++(3, 1);
        \fillc{14}{-2}{1}

       \draw[step=1cm, black, thick] (6, 1) grid (16, -2);
       \draw[step=1cm, black, ultra thick, ->] (4.75,1) -> (4.75,-2); 
        \node[rotate=90] at (4.25, -0.5) {\Huge {Time}};
        \node [] at (10.5, -2.7) {\Huge {$\,{} ^{\omega}\bot aabc\bot^{\omega}\, {\in}\, I$}};
        \end{tikzpicture}
    }}
   \hspace{1mm}
\subfloat[$\{a^n b^m \,|\, n\,{\geq}\,m\,{>}\,0\}$\label{fig:OCA-lang1-d}]{
  \centering
        \scalebox{0.9025}{$ \begin{array}{|l|l|}
            \hline
            \mybot  & \mybot \mybot \mybot  \\[0.3mm] \hline
            a & \text{\makebox[\widthof{$\mybot \mybot \mybot$}][c]{$\mybot \mybot a$}} \quad  \mybot aa \quad \text{\makebox[\widthof{$b\mybot \mybot$}][c]{$\mybot ab$}} \quad %
             \text{\makebox[\widthof{$bb\mybot$}][c]{$aaa$}} \quad aab \\[0.3mm] \hline %
            b & \text{\makebox[\widthof{$\mybot \mybot \mybot$}][c]{$ab\mybot$}}  \quad \text{\makebox[\widthof{$\mybot aa$}][c]{$abb$}}  \quad  b\mybot \mybot  \quad bb\mybot  \quad \text{\makebox[\widthof{$aab$}][c]{$bbb$}} \\[0.2mm]
            \hline
        \end{array}$}\hspace{1mm}
    \scalebox{0.3525}{
     \begin{tikzpicture}[baseline={(current bounding box.center)}]
     \draw[step=1cm, gray, thin] (6, 1) grid (16, -2);
        \foreach \i in {1,2,3} {
            \node at (5.5, 1.5-\i) {\huge ${\ldots}$};
            \node at (16.5,  1.5-\i) {\huge ${\ldots}$};
        }
        \fill[lightgray!50] (9, 0) rectangle ++(3, 1);
        \fill[darkgray!50] (12, 0) rectangle ++(1, 1);
        
        \fill[lightgray!50] (8, -1) rectangle ++(4, 1);
        \fill[darkgray!50] (12, -1) rectangle ++(2, 1);

        \fill[lightgray!50] (7, -2) rectangle ++(5, 1);
        \fill[darkgray!50] (12, -2) rectangle ++(3, 1);

       \draw[step=1cm, black, thick] (6, 1) grid (16, -2);
       \draw[step=1cm, black, ultra thick, ->] (4.75,1) -> (4.75,-2); 
        \node[rotate=90] at (4.25, -0.5) {\Huge {Time}};
        \node [] at (10.5, -2.7) {\Huge {$\,{} ^{\omega}\bot aaa b\bot^{\omega} \, {\in}\, I$}};
        \end{tikzpicture}
    }
}
    \caption[]{
  {Legend:} $\bot$ {=} $\whitebox$, $a$ {=} $\abox$, $b$ {=} $\bbox$, $c$ {=} $\cbox$
  }\label{fig:OCA-lang1}
\end{figure}
    
    \begin{restatable}{proposition}{OCAproposition}\label{OCA-proposition}
         For every pair of consecutive classes $\class{C}$ and $\class{C}'$ in Thm.\ref{Thm:one-counter-hierarchy} there exists $L\,{\in}\,\class{C}'\,{\setminus}\class{C}$ that is expressible by a regular initialized CA, i.e., $L\,{\in}\,\caclassshortREG{1}{r}{\Sigma}{REG}$.
    \end{restatable}
\vspace{-1.85mm}
     \begin{restatable}{claim}{voclReg}\label{claim:CAvocl} 
    There exists $L \,{\in}\,{{\class{VOCL}}\,{\setminus}\,{\class{REG}}} $ such that $  \caclassshortREG{1}{r}{\Sigma}{REG}$.
    \end{restatable}
\vspace{-2.85mm}
    \begin{proof}
    The language $L\,{=}\,\{a^n b^n \,|\, n\,{\in}\,\mathbb{N}\} \,{\in}\, \class{VOCL}\setminus\class{REG}$. 
    Let $\mathcal{A} \in \caclassshort{1}{1}{\{a,b,\bot\}}$ and $I\,{=}\,\ppad(\{ab\})$.
    Fig.\ref{fig:OCA-lang1-a} shows $f$ such that $\mathcal{L}(\mathcal{A},I) \,{=}\, L$. Thus, $L\,{\in}\, \caclassshortREG{1}{1}{\{a,b,\bot\}}{REG}$.
    \end{proof}
    \vspace{-1.95mm}
     \begin{claim}\label{claim:CArocl}
        There exists $L\, {\in}\, \class{ROCL}\setminus\class{VOCL} $ such that $ \caclassshortREG{1}{r}{\Sigma}{REG}$. 
    \end{claim}
\vspace{-2.85mm}
    \begin{proof}
        The language $L\,{=}\,\{a^n b a ^n |\, n\,{\in}\,\mathbb{N}\} \,{\in}\, \class{ROCL}{\setminus}\class{VOCL}$. 
    Let $\mathcal{A} \, {\in}\, \caclassshort{1}{1}{\{a,b,\bot\}}$\\
    and $I \, {=}\,  \ppad(\{aba\})$. Fig.\ref{fig:OCA-lang1-b} shows $f$ 
    s.t. $\mathcal{L}(\mathcal{A},I)\,{=}\, L$. Thus, $L\,{\in}\, \caclassshortREG{1}{1}{\{a,b,\bot\}}{REG}$.
    \end{proof}
\vspace{-1.95mm}
    \begin{claim}\label{claim:CAdocl}
        There exists $L \,{\in}\, \class{DOCL}{\setminus}\class{ROCL}$ such that $ \caclassshortREG{1}{r}{\Sigma}{REG}$.
    \end{claim} 
\vspace{-2.85mm}
    \begin{proof}
     $L\,{=}\,\{ a^n b^m c \,|\, n\,{\geq}\,m \,{>}\,0\} \,{\in}\,\class{DOCL}{\setminus}\class{ROCL}$. 
    Let $\mathcal{A}\,{\in}\, \caclassshort{1}{1}{\{a,b,c,\bot\}}$ and $I =\ppad(\{ {{a^*} a b c}\})$.
    Fig.\ref{fig:OCA-lang1-c} shows $f$ such that $\mathcal{L}(\mathcal{A},I){=} L$.
    Thus, $L\,{\in}\,\caclassshortREG{1}{1}{\{a,b,c,{\bot }\}}{REG}$
    \end{proof}
    \vspace{-1.85mm}
     \begin{claim}
            There exists $L\,{\in}\,\class{OCL}{\setminus}\allowbreak\class{DOCL}$ such that $ \caclassshortREG{1}{r}{\Sigma}{REG}$.
    \end{claim}
   \vspace{-2.85mm}
    \begin{proof}
    The language $L\,{=}\,\{a^nb^m \,|\, n{\geq}\, m{>}0\}\,{\in}\,\class{OCL}{\setminus}\class{DOCL}$. 
    Let $\mathcal{A}\,{\in}\,\caclassshort{1}{1}{\{a,b,\bot\}}$ and $I\,{=}\,\ppad(\{a^* ab\})$. Fig.\ref{fig:OCA-lang1-d} shows $f$
    s.t.
    $\mathcal{L}(\mathcal{A},I)\,{=}\, L$. Thus, $L\,{\in}\, \caclassshortREG{1}{1}{\{a,b,{\bot}\}}{REG}$.
    \end{proof}
    
    To show that cellular automata can express languages beyond the OCA hierarchy, we introduced the notion of gliders, defined next.

    \section{Gliders Perspective}\label{subsec:gliders}
    In this section, we define \term{gliders} within regular-initialized CA and formalize their interactions and semantic role.     
    Gliders offer a symbolic structure that reflects the emergent behavior of CA dynamics, especially when the rule function is viewed as being shaped by dominant or persistent symbolic patterns.
    Recall from the introduction that our definition of \emph{gliders} differ from the conventional one. Specifically, we regard a glider as a one-cell (\emph{atomic}) entity moving at a fixed velocity, as we formally define below.
    
    \subsection{Formal Definition of Gliders}
     Given a neighborhood ${w \,{\in}\, {\SigmaWindow}}\!$, we index it as $w \,{=}\, \sig{{-}r} \ldots \sig{{-}1}\sig{0}\sig{1} \ldots \sig{r}$.

     \vspace{-1mm}
     
    \paragraph{Notation.}
    For any $\sigma \,{\in}\, \Sigma$ and $i \,{\in}\, [-r, r]$, we define:
    \begin{center}
        $\LocAt{i}{\sigma} \,{:=}\, \{ w \,{\in}\,\SigmaWindow \mid \sig{-i} \,{=}\, \sigma \}$
    \end{center}
     That is, $\LocAt{i}{\sigma}$ is the set of neighborhoods where the symbol $\sigma$ appears in the $-i$ position.\footnote{This corresponds to the classical \emph{cylinder set} notion, typically denoted $[\sigma]_{-i}$.} 
     The parameters $\Sigma$ and $r$ are those of the CA $\mathcal{A} \,{\in}\, \caclassshort{1}{r}{\Sigma}$ currently in context, and are omitted from the notation for brevity.     

    This shorthand will be used extensively in what follows.

    A glider $\g{\sigma}{i}$ is a pair $(\sigma,i)$, where $\sigma$ denotes its \term{value} and $i\,{\in}\,\mathbb{Z}$ denotes its \term{velocity}. 
    The velocity $i$ determines both its \term{speed}, given by $|i|$, and its \term{direction}, given by $\sign(i)$. Positive values indicate motion to the right, negative values indicate motion to the left, and $i\,{=}\,0$ corresponds to no horizontal movement, i.e., the object is static. The formal definition follows.

    \begin{definition}[Glider]\label{def:glider}
    Let ${\mathcal{A}\,{\in}\, \caclassshort{1}{r}{\Sigma}}$.
    For $i\,{\in}\,[{-}r,r]$ and $\sigma\,{\in}\,\Sigma$, we say that $\g{\sigma}{i}$ is a \emph{glider} in $\mathcal{A}$ if there exists $w\,{\in}\,\LocAt{i}{\sigma}$ such that ${{f(w)}}\,{{=}}\,{\sigma}$.
    \end{definition}
    Thus, given $\mathcal{A}\,{\in}\, \caclassshort{1}{r}{\Sigma}$, the gliders of $\mathcal{A}$ are subsumed in the following set: ${\{{\g{\sigma}{i}} \mid { i\,{\in}\,{[{-}r,r]},\ \allowbreak \sigma\in\Sigma}}\}$. 
    Consequently, the total number of distinct gliders admitted by $\mathcal{A}\,{\in}\, \caclassshort{1}{r}{\Sigma}$ is at most $|\Sigma|\,{\cdot}\,(2r{+}1)$. 
 
    Next, we define various types of gliders, the interactions that may arise between them, and the resulting effects on the system, together with behaviors that may or may not be exhibited.

 \begin{definition}[Persistent Glider]
   Let ${\g{\sigma}{i}}$ be a glider in $\mathcal{A}\,{\in}\, \caclassshort{1}{r}{\Sigma}$. We say that ${\g{\sigma}{i}}$ is \term{persistent} if ${f(w)}\,{=}\, \sigma$ for \overlapuline[0.30ex]{0.4pt}{every} $w \,{\in}\,\LocAt{i}{\sigma}$. We say that ${\g{\sigma}{i}}$ is \term{non-persistent} if ${f(w)}\,{=}\, \sigma$ for \overlapuline[0.30ex]{0.4pt}{some but not all} of $w \,{\in}\,\LocAt{i}{\sigma}$.
   
 \end{definition}
    Note that, if ${\g{\sigma}{i}}$ is \emph{persistent} in $\mathcal{A}$ then any configuration $c\,{\in}\,\SigmaZWindow$ satisfies that ${\langle z\rangle_{\scriptstyle c}\,{=}\,\sigma}$ implies ${\langle z{+}i\rangle_{\scriptstyle{G(c)}}\,{=}\,\sigma} $ for any $ z\,{\in}\,\mathbb{Z}$.
\paragraph{Spreading State.}
    Let $ m, n \,{\in}\, \mathbb{Z} $ for $ m \,{<}\, n $, $\sigma\,{\in}\, \Sigma$, and $[m,n] \subseteq  {[{-}r,r]}$.
    Assume that $\mathcal{A}\,{\in}\, \caclassshort{1}{r}{\Sigma}$ implements a set of \emph{persistent gliders}, $\{\g{\sigma}{i} \mid i\in [m,n]\}$, i.e., all \emph{persistent gliders} of value ${\sigma}$ and velocity $i$ for any $i\,{\in}\,[m,n]$.
    Then we say that $\mathcal{A}$ has a \emph{$\sigma$-spreading state}.
     We consider the following types of $\sigma$-spreading state:
    \begin{enumerate}[i.,topsep=1.5mm,itemsep=1.5mm]
        \item 
         $m\,{=}\, 0$: Then $\mathcal{A}$ has a \term{right}-spreading state (at speed $n$) (Fig.\ref{fig:right-spreading-state}).
        \item 
        \makebox[\widthof{$m$}][l]{$n$}$\,{=}\, 0$: Then $\mathcal{A}$ has a \term{left}-spreading state (at speed $|m|$) (Fig.\ref{fig:left-spreading-state}).
        \item $0\,{\in}\,(m,n)$: Then $\mathcal{A}$ has a \term{two-way} spreading-state (Fig.\ref{fig:two-way-spreading-state}).
        \item \makebox[\widthof{$0\,{\in}\,(m,n)$}][l]{$0\,{\notin}\,[m,n]$}:
       Then $\mathcal{A}$ has a \term{shift}-spreading state (Fig.\ref{fig:shift-spreading-state}). 
    \end{enumerate} 
    Intuitively, a $\sigma$-spreading state enables a prediction of the orbit behavior of a $c\,{\in}\,\SigmaZWindow$, based solely on the fact that $\langle i \rangle_{\scriptstyle c} \,{=}\,\sigma$. 
    E.g., having a two-way $\sigma$-spreading state (resp. right / left spreading), means that eventually every cell of $c$ in positions $\mathbb{Z}$ (resp. every cell in position $z \,{\geq}\,i$ or $z \,{\leq}\, i$) would be equal to $\sigma$.

    \begin{figure}[hb!]
    \centering
    \vspace{-7.5mm}
  \subfloat[Right-spread (speed 2)\label{fig:right-spreading-state}]{  
     \begin{tikzpicture}[scale=0.395, every node/.style={scale=0.395}]
        \foreach \i in {1,2,3,4} {
            \fill[lightgray!50] (0, -\i+1) rectangle ++(\i+\i-1, 1);
            \fill[pattern=north west lines, pattern color=lightgray] (0, -\i+1) rectangle ++(\i+\i-1, 1);
            \node at (- 1.5, 1.5-\i) {\huge ${\ldots}$};
            \node at (8.5,  1.5-\i) {\huge ${\ldots}$};
        }
        \draw[step=1cm, darkgray, thin] (-1, 1) grid (8, -3);
        \blackonwhiteScaled{1.95}{-1.35}{\gfig{\boldsymbol{\sigma}}{\boldsymbol{{+}1}}}{\huge}
        \draw[thick,black,->] (0.5,0.5) -- node[below right=10pt, yshift=-5pt, xshift=-11pt] {~} (1.25,-0.5); 
        \draw[thick,black,->]
        (0.5,0.5) -- node[below=10pt, yshift=4pt] {~} (0.5,-0.5); 
         \blackonwhiteScaled{0}{-1.35}{\gfig{\boldsymbol{\sigma}}{\boldsymbol{0}}}{\huge}
        \draw[thick,black,->]
        (0.5,0.5) -- node[below right=10pt, xshift=30pt, yshift=-5pt] {~} (2.5,-0.5); 
        \blackonwhiteScaled{4.15}{-1.35}{\gfig{\boldsymbol{\sigma}}{\boldsymbol{{+}2}}}{\huge}
        \draw[step=1cm, black, semithick, ->] (-2.25,1) -> (-2.25,-3); 
        \node[rotate=90] at (-2.75, -1) {\Huge {Time}};
        \end{tikzpicture}
    }
     \hspace{4mm}
   \subfloat[Two-way-spread\label{fig:two-way-spreading-state}]{\small 
     \begin{tikzpicture}[scale=0.395, every node/.style={scale=0.395}]
        \foreach \i in {1,2,3,4} {
            \fill[lightgray!50] (11, -\i+1) rectangle ++(-\i, 1);
            \fill[lightgray!50] (10, -\i+1) rectangle ++(\i, 1);
            \fill[pattern=north west lines, pattern color=lightgray](11, -\i+1) rectangle ++(-\i, 1);
            \fill[pattern=north west lines, pattern color=lightgray](10, -\i+1) rectangle ++(\i, 1);
            \node at (5.5, 1.5-\i) {\huge ${\ldots}$};
            \node at (15.5,  1.5-\i) {\huge ${\ldots}$};
        }
        \draw[step=1cm, darkgray, thin] (6, 1) grid (15, -3);
        \draw[thick,black,->]
        (10.5,0.5) -- node[below left=12pt, yshift=-1pt] {~} (9.5,-0.5);
        \draw[thick,black,->]
        (10.5,0.5) -- node[below=12pt, yshift=2pt] {~} (10.5,-0.5); 
        \draw[thick,black,->]
        (10.5,0.5) -- node[below right=12pt, yshift=-1pt] {~} (11.5,-0.5); 
         \blackonwhiteScaled{8.45}{-1.35}{\gfig{\boldsymbol{\sigma}}{\boldsymbol{{-}1}}}{\huge}

         \blackonwhiteScaled{10.5}{-1.35}{\gfig{\boldsymbol{\sigma}}{\boldsymbol{0}}}{\huge}

         \blackonwhiteScaled{12.55}{-1.35}{\gfig{\boldsymbol{\sigma}}{\boldsymbol{{+}1}}}{\huge}
        \draw[step=1cm, black, semithick, ->] (4.75,1) -> (4.75,-3); 
        \node[rotate=90] at (4.25, -1) {\Huge {Time}};
        \end{tikzpicture}
    }

     \hspace{0mm} 
     \subfloat[Left-spread (speed 1)\label{fig:left-spreading-state}]{\small 
     \begin{tikzpicture}[scale=0.395, every node/.style={scale=0.395}]
        \foreach \i in {1,2,3,4} {
            \fill[lightgray!50] (11, -\i+1) rectangle ++(-\i, 1);
             \fill[pattern=north west lines, pattern color=lightgray](11, -\i+1) rectangle ++(-\i, 1);
            
            \node at (3.5, 1.5-\i) {\huge ${\ldots}$};
            \node at (13.5,  1.5-\i) {\huge ${\ldots}$};
        }
         \draw[step=1cm, darkgray, thin] (4, 1) grid (13, -3);
        \draw[thick,black,->]
        (10.5,0.5) -- node[below left=10pt, yshift=-4.5pt] {~} (9.5,-0.5);
         \blackonwhiteScaled{8.45}{-1.35}{\gfig{\boldsymbol{\sigma}}{\boldsymbol{{-}1}}}{\huge}

         \blackonwhiteScaled{10.5}{-1.35}{\gfig{\boldsymbol{\sigma}}{\boldsymbol{0}}}{\huge}
        \draw[thick,black,->]
        (10.5,0.5) -- node[below =10pt, yshift=4pt] {~} (10.5,-0.5);
        \draw[step=1cm, black, semithick, ->] (2.75,1) -> (2.75,-3); 
        \node[rotate=90] at (2.25, -1) {\Huge {Time}};
        \end{tikzpicture}
    } \hspace{4mm}
        \subfloat[Shift-spread\label{fig:shift-spreading-state}]{\small 
     \begin{tikzpicture}[scale=0.395, every node/.style={scale=0.395}]
        \foreach \i in {1,2,3,4} {
            \node at (5.5, 1.5-\i) {\huge ${\ldots}$};
            \node at (15.5,  1.5-\i) {\huge ${\ldots}$};
        }
        \fill[lightgray!50] (7, 0) rectangle ++(1, 1);
         \fill[pattern=north west lines, pattern color=lightgray](7, 0) rectangle ++(1, 1);
        \fill[lightgray!50] (8, -1) rectangle ++(2, 1);
        \fill[pattern=north west lines, pattern color=lightgray](8, -1) rectangle ++(2, 1);
        \fill[lightgray!50] (9, -2) rectangle ++(3, 1);
        \fill[pattern=north west lines, pattern color=lightgray](9, -2) rectangle ++(3, 1);
        \fill[lightgray!50] (10, -3) rectangle ++(4, 1);
        \fill[pattern=north west lines, pattern color=lightgray](10, -3) rectangle ++(4, 1);
        \draw[step=1cm, darkgray, thin] (6, 1) grid (15, -3);
        \draw[thick,black,->]
        (7.5,0.5) -- node[below =10pt, yshift=4pt, xshift=10pt] {~} (8.5,-0.5);
        \draw[thick,black,->]
        (7.5,0.5) -- node[below=10pt, yshift=4pt, xshift=55pt] {~} (9.5,-0.5);
         \blackonwhiteScaled{11.75}{-1.35}{\gfig{\boldsymbol{\sigma}}{\boldsymbol{{+}2}}}{\huge}

         \blackonwhiteScaled{9.35}{-1.35}{\gfig{\boldsymbol{\sigma}}{\boldsymbol{+1}}}{\huge}
         \draw[step=1cm, black, semithick, ->] (4.75,1) -> (4.75,-3); 
        \node[rotate=90] at (4.25, -1) {\Huge {Time}};
        \end{tikzpicture}
    }
    
    \caption{Examples of spreading state behavior} %
    \label{fig:spreading-state}
    \vspace{-10mm}
    \end{figure}
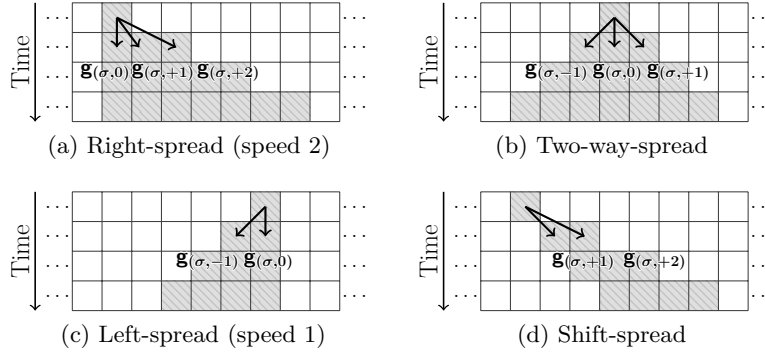

\subsubsection{Coexistence of Persistent Gliders}\label{sec:collision}
Not all persistent gliders can coexist.
Since each glider is defined by a value and velocity, four pairwise cases arise.
Let $\mathcal{A}\,{\in}\,\caclassshort{1}{r}{\Sigma}$ and let $\g{\sigma}{i}, \g{\tau}{j}$ be two persistent gliders in $\mathcal{A}$.
\paragraph{Same velocity.} If $i\,{=}\,j$, the gliders never collide. If, in addition, $\sigma\,{=}\,\tau$, then they are identical. If $\sigma\,{\neq}\,\tau$, then $\LocAt{i}{\sigma}\cap\LocAt{i}{\tau} \,{=}\,\emptyset$, as every cell admits a single value, so no conflict arises. (See Figs.\ref{fig:colision-new-a}, \ref{fig:colision-new-c})
\paragraph{Different velocities, same value.}\label{item:persist-2} 
        If $i\,{\neq}\, j$ and $\sigma\,{=}\,\tau$, then the gliders may \emph{collide}, but agree on the resulting value. 
       That is, for any $w\,{\in}\,\LocAt{i}{\sigma}\,{\cup}\, \LocAt{j}{\sigma}$, the rule function satisfies both by assigning $f(w)\,{=}\,\sigma$. Hence, no conflict arises. (Fig.\ref{fig:colision-new-b})
\paragraph{Different velocities, different values.} 
    If ${i\,{\neq}\, j}$ and ${\sigma\, {\neq}\, {\tau}}$, then $\LocAt{i}{\sigma}  \cap   \LocAt{j}{\tau}  \,{\neq}\, \emptyset$. Any $w\,{\in}\,\LocAt{i}{\sigma}  \cap   \LocAt{j}{\tau}$ must satisfy both ${f(w)\,{=}\,\sigma}$ and ${f(w)} \,{=} \,{\tau}$. However, $\sigma\,{\neq}\,\tau$. Since a single cell cannot hold multiple values at once, this leads to contradiction. Thus such persistent gliders cannot coexists. (Fig.\ref{fig:colision-new-d})

    \begin{figure}[ht!]
    \centering
     \subfloat[Same Value, Same Velocity \label{fig:colision-new-a}]{
        \scalebox{0.39}{
     \begin{tikzpicture}
        \foreach \shift in {10} {
        \foreach \i in {1,2,3} {
           \fill[lightgray!50] (-1+\i, -\i+1) rectangle ++(1, 1);

            \fill[lightgray!50] (2+\i, -\i+1) rectangle ++(1, 1);
            \node at (- 2.75, 1.5-\i) {\huge ${\ldots}$};
            \node at (17.75-\shift,  1.5-\i) {\huge ${\ldots}$};
        }
        \draw[step=1cm, black, thick] (8-\shift, 1) grid (17-\shift, -2);
        }
        \foreach \j in {0,1}{
        \pgfmathparse{
        ifthenelse(\j==0,"black",
        ifthenelse(\j==1,"gray",
        "lightgray"))}
        \edef\mycolor{\pgfmathresult}
          \foreach \i in {0.5+\j, 3.5+\j} {
            \draw[ultra thick,\mycolor, dashed, ->, shorten >=3pt, shorten <=2pt] (\i,0.5-\j) -- node[below right] {~} (\i+1,-0.5-\j);
        }}
        \draw[step=1cm, black, very thick, ->] (-3.65,1) -> (-3.65,-2); 
        \node[rotate=90] at (-4.2, -0.7) {\Huge {Time}};
        \end{tikzpicture}
    }}
\hspace{5mm}
\subfloat[Different Values, Same Velocity \label{fig:colision-new-c}]{
 \scalebox{0.39}{
     \begin{tikzpicture}
     \foreach \shift in {10} {
        \foreach \i in {1,2,3} {
           \fill[lightgray!50] (-1+\i, -\i+1) rectangle ++(1, 1);

            \fill[gray] (2+\i, -\i+1) rectangle ++(1, 1);
            \node at (- 2.75, 1.5-\i) {\huge ${\ldots}$};
            \node at (18.75-\shift,  1.5-\i) {\huge ${\ldots}$};
        }
        \draw[step=1cm, black, thick] (8-\shift, 1) grid (18-\shift, -2);
        }
        \foreach \i in {1} {
            \draw[ultra thick,black, dashed, ->, shorten >=3pt, shorten <=2pt] (-0.5+\i,1.5-\i) -- node[below right] {~} (0.5+\i,0.5-\i);
            \draw[ultra thick,black, dashed,->, shorten >=3pt, shorten <=2pt]
        (2.5+\i,1.5-\i) -- node[below right=12pt, yshift=-1pt, xshift=1pt] {~} (3.5+\i,0.5-\i);
        }
        \foreach \i in {2} {
            \draw[ultra thick,gray, dashed,->, shorten >=3pt, shorten <=2pt] (-0.5+\i,1.5-\i) -- node[below right] {~} (0.5+\i,0.5-\i);
            \draw[ultra thick,gray!50, dashed,->, shorten >=3pt, shorten <=2pt]
        (2.5+\i,1.5-\i) -- node[below right=12pt, yshift=-1pt, xshift=1pt] {~} (3.5+\i,0.5-\i);
        }
        
        \draw[step=1cm, black, very thick, ->] (-3.65,1) -> (-3.65,-2); 
        \node[rotate=90] at (-4.2, -0.7) {\Huge {Time}};
        \end{tikzpicture}
    }}

   \subfloat[Same Value, Different Velocities \label{fig:colision-new-b}]{ \scalebox{0.39}{
    \begin{tikzpicture}
        \foreach \i in {1,2} {
            \fill[lightgray!50] (6+\i, -\i+1) rectangle ++(1, 1);

            \fill[lightgray!50] (9+\i, -\i+1) rectangle ++(1, 1);
             \fill[lightgray!50] (11-\i, -\i+1) rectangle ++(1, 1);
            \fill[lightgray!50] (10+\i, -\i+1) rectangle ++(1, 1);
             \fill[lightgray!50] (12-\i, -\i+1) rectangle ++(1, 1);
             \fill[lightgray!50] (8-\i, -\i+1) rectangle ++(1, 1);
        }
        \foreach \i in {1,2,3,4} {
            \node at (5.25, 1.5-\i) {\huge ${\ldots}$};
            \node at (15.75,  1.5-\i) {\huge ${\ldots}$};
        }
        \fill[lightgray!50] (7, -2) rectangle ++(7, 1);
        \fill[lightgray!50] (6, -3) rectangle ++(9, 1);
         \draw[step=1cm, black, thick] (6, 1) grid (15, -3);
         \foreach \i in {7.5}
         {
             \draw[ultra thick,black, dashed, ->, shorten >=3pt, shorten <=2pt]
            (\i,0.5) -- node[below left=14pt, yshift=0.5pt] {~} (\i-1,-0.5);
            \draw[ultra thick,black, dashed, ->, shorten >=3pt, shorten <=2pt]
            (\i,0.5) -- node[below right =12pt, yshift=0.5pt] {~} (\i+1,-0.5);
            \foreach \j in {1} {
            \draw[ultra thick,gray, dashed, ->, shorten >=3pt, shorten <=2pt]
            (\i+\j,0.5-\j) -- node[below left=12pt, yshift=0.5pt] {~} (\i+\j-1,-0.5-\j);
            \draw[ultra thick,gray, dashed, ->, shorten >=3pt, shorten <=2pt]
            (\i+\j,0.5-\j) -- node[below right=12pt, yshift=0.5pt] {~} (\i+\j+1,-0.5-\j);
            \draw[ultra thick,gray, dashed, ->, shorten >=3pt, shorten <=2pt]
            (\i-\j,0.5-\j) -- node[below right =12pt, yshift=0.5pt] {~} (\i-\j+1,-0.5-\j);
            \draw[ultra thick,gray, dashed, -, shorten >=3pt, shorten <=2pt]
            (\i-\j,0.5-\j) -- node[below left =12pt] {~} (\i-\j-1,-0.5-\j);
            }
            \foreach \j in {2} {
            \draw[ultra thick,gray, dashed, ->, shorten >=3pt, shorten <=2pt]
            (\i+\j,0.5-\j) -- node[below left=12pt, yshift=0.5pt] {~} (\i+\j-1,-0.5-\j);
            \draw[ultra thick,gray, dashed, ->, shorten >=3pt, shorten <=2pt]
            (\i+\j,0.5-\j) -- node[below right=12pt, yshift=0.5pt] {~} (\i+\j+1,-0.5-\j);
            }
         }
        \foreach \i in {10.5, 11.5}
         {
             \draw[ultra thick,black, dashed, ->, shorten >=3pt, shorten <=2pt]
            (\i,0.5) -- node[below left=14pt, yshift=0.5pt] {~} (\i-1,-0.5);
            \draw[ultra thick,black, dashed, ->, shorten >=3pt, shorten <=2pt]
            (\i,0.5) -- node[below right =12pt, yshift=0.5pt] {~} (\i+1,-0.5);
            \foreach \j in {1} {
            \draw[ultra thick,gray, dashed, ->, shorten >=3pt, shorten <=2pt]
            (\i+\j,0.5-\j) -- node[below left=12pt, yshift=0.5pt] {~} (\i+\j-1,-0.5-\j);
            \draw[ultra thick,gray, dashed, ->, shorten >=3pt, shorten <=2pt]
            (\i+\j,0.5-\j) -- node[below right=12pt, yshift=0.5pt] {~} (\i+\j+1,-0.5-\j);
            \draw[ultra thick,gray, dashed, ->, shorten >=3pt, shorten <=2pt]
            (\i-\j,0.5-\j) -- node[below right =12pt, yshift=0.5pt] {~} (\i-\j+1,-0.5-\j);
            \draw[ultra thick,gray, dashed, ->, shorten >=3pt, shorten <=2pt]
            (\i-\j,0.5-\j) -- node[below left =12pt, yshift=0.5pt] {~} (\i-\j-1,-0.5-\j);
            }
            \foreach \j in {2} {
                \foreach \r in {0,-1, -2, -3, -4 ,-5}{
            \draw[ultra thick,white!50, dashed, ->, shorten >=3pt, shorten <=2pt]
            (\i+\j+\r,0.5-\j) -- node[below left=12pt, yshift=0.5pt] {~} (\i+\j+\r-1,-0.5-\j);
            \draw[ultra thick,white!50, dashed, ->, shorten >=3pt, shorten <=2pt] (\i+\j+\r,0.5-\j) -- node[below right=12pt, yshift=0.5pt] {~} (\i+\j+\r+1,-0.5-\j);}
            }
         }
        \draw[step=1cm, black, very thick, ->] (4.35,1) -> (4.35,-3); 
        \node[rotate=90] at (3.8, -1) {\Huge {Time}};
        \end{tikzpicture}
    }}
   \hspace{5mm}
   \subfloat[Different Values, Different Velocities \label{fig:colision-new-d}]{
        \scalebox{0.39}{
     \begin{tikzpicture}
     \foreach \shift in {10} {
        \foreach \i in {1,2,3} {
           \fill[lightgray!50] (-1+\i, -\i+1) rectangle ++(1, 1);
            \fill[lightgray!50] (4+\i, -\i+1) rectangle ++(1, 1);
            \fill[gray] (5-\i, -\i+1) rectangle ++(1, 1);
            \node at (- 1.75, 1.5-\i) {\huge ${\ldots}$};
            \node at (19.75-\shift,  1.5-\i) {\huge ${\ldots}$};
        }
        \fill[lightgray!50] (2,-2) rectangle +(1,1);
        \fill[gray] (2,-2) -- (3,-2) -- (3,-1) -- cycle; 
        \draw[step=1cm, black, thick] (9-\shift, 1) grid (19-\shift, -2);
        }
        \foreach \j in {0,1}{
        \pgfmathparse{
        ifthenelse(\j==0,"black", "gray")}
        \edef\mycolor{\pgfmathresult}
          \foreach \i in {0.5+\j, 5.5+\j} {
            \draw[ultra thick,\mycolor, dashed, ->, shorten >=3pt, shorten <=2pt] (\i,0.5-\j) -- node[below right] {~} (\i+1,-0.5-\j);
        }}
        \foreach \j in {0,1}{
        \pgfmathparse{
        ifthenelse(\j==0,"black", "white")}
        \edef\mycolor{\pgfmathresult}
          \foreach \i in {4.5-\j} {
            \draw[ultra thick,\mycolor, dashed, ->, shorten >=3pt, shorten <=2pt] (\i,0.5-\j) -- node[below left] {~} (\i-1,-0.5-\j);
        }}
        \blackonwhite{2.5}{-1.5}{?}
        \node at (2.5,-2.85) {\huge\textbf{Contradiction!}};
         \draw[step=1cm, black, very thick, ->] (-2.65,1) -> (-2.65,-2); 
        \node[rotate=90] at (-3.2, -0.7) {\Huge {Time}};
        \end{tikzpicture}
    }}
    
    \vspace{-4mm}
    
        \caption{Coexistence of Persistent Gliders} \label{fig:colision-new}
         \vspace{-4mm}
    
    \end{figure}
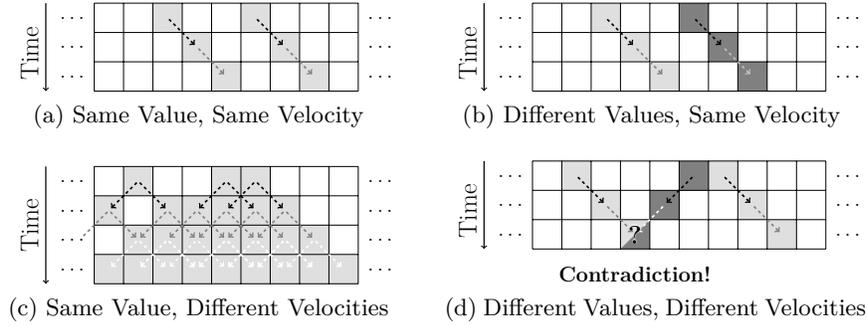
    \begin{restatable}[Coexistence of Persistent Gliders]
    {corollary}{persist}\label{cor:persist-glid}
         If $\g{\sigma}{i}$ and $\g{\tau}{j}$ are both persistent gliders in $\mathcal{A}$, then either $i \,{=}\, j$ or $\sigma \,{=}\, {\tau}$.
    \end{restatable}
    Thus, the set $\G$ of persistent gliders must consist of gliders that either share the same velocities or carry the same value.
  \begin{note}\label{obs:persistent}
Two limiting cases illustrate the system-wide effect of persistent gliders. 
Let ${\mathcal{A}\,{\in}\, \caclassshort{1}{r}{\Sigma}}$, ${i\,{\in}\,[{-}r,r]}$, and consider the set ${\G_i {=} \{ \g{\sigma}{i} \,|\,{{\sigma {\in} \Sigma}}\}}$ of persistent gliders. 
    \begin{enumerate}[i.,topsep=1.5mm,itemsep=1.5mm]
        \item If $\mathcal{A}$ implements all gliders in $\G_i$, then any ${c\,{\in}\,\SigmaZWindow}$ shifts uniformly by $i$. That is, for any ${z\in\mathbb{Z}}$, it holds that ${G(c)[z]\,{=}\, c[z{-} i]}$.
    Hence, ${\orbitc{I} {=} I}$ and ${\mathcal{L}(\mathcal{A},F) {=}F}$ for any ${I\,{\subseteq}\,\SigmaZWindow}$ and $F\,{\subseteq}\,\Sigma^*$.

    \item If $\mathcal{A}$ implements all gliders in ${\G_i \,{\setminus}\, \{\g{\mybot}{i}\}}$, then for any $w\,{\in}\allowbreak (\Sigma\,{\setminus}\{\mybot\})^*$ and every infix ${w'}\sqsubset w$, it holds that $w'\sqsubset c'$ of every $c'\,{\in}\, \orbitc{\ppad(\{w\})}$.
    \end{enumerate}
\end{note}   
These cases show how implemented persistent gliders constrain global behavior-ranging from uniform shifts to infix preservation.

    Recall that a CA language is defined with respect to a set $I\,{\subseteq}\, C_{\mathcal{F}}$ of initial configurations.
    Furthermore, in this setting, the local rule function, $f$, may be partial, defined only for the feasible neighborhoods, namely, $\FN_{\mathcal{A},I}$. Accordingly, we define the notion of an \emph{$I$-persistent glider} as follows.
 \begin{definition}[$I$-Persistent Gliders]
       Let $\mathcal{A}\,{\in}\,\caclassshort{1}{r}{\Sigma}$ and $I\,{\subseteq}\, C_{\mathcal{F}}$. 
       A glider $\g{\sigma}{i}$ is \term{$I$-persistent} in $\mathcal{A}$ if for all $w\,{\in}\, \LocAt{i}{\sigma} \cap \FN_{\mathcal{A},I}$, we have $f(w) \,{=}\, \sigma$.
 \end{definition}
       Note that $I$-persistence generalizes persistence, as every {persistent-glider} is $I$-persistent for any $I$, but not vice versa. 
       The following example illustrates the concept of being persistent with respect to a particular initial configuration $I$.
\begin{example}[$I$-Persistent Gliders]\normalfont
        Let $\mathcal{A}\,{\in}\, \caclassshort{1}{1}{\{a,b,\bot \}}$ and $f$ defined as follows.
        \begin{center}
            \scalebox{0.885}{$ \begin{array}{|l|l|}
                \hline
                \mybot  & \mybot \mybot \mybot  \quad \text{\makebox[\widthof{$\mybot a \mybot$}][c]{$a\mybot b$}} \quad \text{\makebox[\widthof{$a \mybot \mybot$}][c]{$aa\mybot$}}  \quad \text{\makebox[\widthof{$a \mybot a$}][c]{$aba$}} \quad \text{\makebox[\widthof{$ \mybot b \mybot$}][c]{$b\mybot b$}} \\[0.15mm]\hline
                a & \text{\makebox[\widthof{$\mybot \mybot \mybot$}][c]{$\mybot \mybot a$}} \quad \mybot a\mybot  \quad \text{\makebox[\widthof{$a \mybot \mybot$}][c]{$\mybot aa$}} \quad \mybot ab \quad \mybot b\mybot  \\
                & \text{\makebox[\widthof{$\mybot \mybot \mybot$}][c]{$\mybot ba$}} \quad \text{\makebox[\widthof{$\mybot a \mybot$}][c]{$aaa$}} \quad \text{\makebox[\widthof{$a\mybot \mybot$}][c]{$aab$}} \quad \text{\makebox[\widthof{$\mybot a b$}][c]{$ba\mybot$}}  \quad \text{\makebox[\widthof{$\mybot b \mybot$}][c]{$bba$}} \\[0.15mm]\hline
                b & \text{\makebox[\widthof{$\mybot \mybot \mybot$}][c]{$\mybot \mybot b$}} \quad \text{\makebox[\widthof{$\mybot a \mybot$}][c]{$\mybot bb$}} \quad a\mybot \mybot  \quad a\mybot a \quad ab\mybot  \quad abb \\
                & \text{\makebox[\widthof{$\mybot \mybot \mybot$}][c]{$b\mybot \mybot$}}  \quad \text{\makebox[\widthof{$\mybot a \mybot$}][c]{$b\mybot a$}}\quad \text{\makebox[\widthof{$a \mybot \mybot$}][c]{$baa$}} \quad \text{\makebox[\widthof{$a \mybot a$}][c]{$bab$}} \quad\,  bb\mybot  \quad bbb \\
                \hline
            \end{array}$}
        \end{center}
       \noindent
       Let $I  \,{=}\, \ppad( \{ ab\})$, so $f$
       is defined as follows (corresponds to Clm.\ref{claim:CAvocl}).
          \begin{center}
               \scalebox{0.885}{$ \begin{array}{|l|l|}
           \hline
            \mybot  & \mybot \mybot \mybot \\[0.15mm] \hline
        a & \mybot \mybot a \quad \mybot aa \quad \text{\makebox[\widthof{$b\mybot \mybot$}][c]{$\mybot a b$}} \quad aaa \quad aab\\[0.15mm] \hline
        b & \text{\makebox[\widthof{$\mybot \mybot a$}][c]{$ab\mybot$}}  \quad  \text{\makebox[\widthof{$\mybot a a$}][c]{$abb$}} \quad b\mybot \mybot  \quad  bb\mybot   \quad bbb \\
            \hline
                \end{array}$}
          \end{center}
        Note that while $\g{a}{\mminus 1}$, $\g{a}{0}$, $\g{b}{0}$ and $\g{b}{+{ } 1}$ are not persistent in $\mathcal{A}$ owing to $f(a\mybot  a) \,{=}\,  b$, $f(aa \mybot) \,{=}\,  \mybot$, $f(\mybot ba) \,{=}\,  a$, $f(b\mybot b) \,{=}\,   \mybot$ resp.
     They are \emph{$I$-persistent}.
\end{example}

\begin{remark}
    Note that \Cref{cor:persist-glid} considers all possible neighborhoods. Under a particular initial configuration $I$, there can be two $I$-persistent gliders $\g{\sigma}{i}, \g{\tau}{j}$, such that both $i\neq j$ and $\sigma\neq \tau$. 
    If this is the case, this is because the contradicting environment is not feasible under $I$.
\end{remark}

\subsection{Gliders Derived by a CA under $I$}\label{sec:derived}
We can now reason not only about individual gliders, but also about their interactions under a fixed initial configuration set.
Given $\mathcal{A}\,{\in}\,\caclassshort{1}{r}{\Sigma}$ and $I\,{\subseteq}\,C_{\mathcal{F}}$, we denote the set of gliders derived from $\mathcal{A}\,|_I$ by $\G_{\mathcal{A}|_I}$.
By Def.\ref{def:glider}, this is the set of gliders $\g{\sigma}{i}$ for which there {exists a} $w\,{\in}\,\FN_{\mathcal{A},I}$ that is in $\LocAt{i}{\sigma}$, namely for which $\sig{-i}\,{=}\,\sigma$.
Some of them might be $I$-persistent, while others might not.

\paragraph{The Derived Partial Order.} Conflicting gliders may overlap in their neighborhood constraints w.r.t. $\mathcal{A}$ and $I \,{\subseteq}\, C_{\mathcal{F}}$. 
To compare such cases, we define a \emph{dominance relation}.
Let $\g{\sigma}{i}, \, \g{\tau}{j} \,{\in}\, \G_{\mathcal{A}|_I}$ with $\sigma \,{\neq}\, \tau$ and $i\,{\neq}\, j$. 
We write 
$$\g{\sigma}{i} \,{>}_{{\mathcal{A}|_{I}}}\ \g{\tau}{j} \text{ ~ iff ~ }  f(w)\, {=}\, \sigma \text{ for all } w\,{\in}\,\LocAt{i}{\sigma} \,{\cap}\, \LocAt{j}{\tau}\,{\cap}\,\FN_{\mathcal{A},I}$$
When the context is clear, we simply write $\g{\sigma}{i} \,{>}\, \g{\tau}{j}$.
%
Intuitively, $\g{\sigma}{i}$ dominates $\g{\tau}{j}$ when the local rule $f$ assigns $\sigma$ to every \FN\ that overlaps.
%
Thus, a dominating glider consistently overrides a weaker one in every overlapping context.
For example, if every $w\,{\in}\, \LocAt{i}{x} \,{\cap}\, \LocAt{j}{y}$ satisfies $f(w)\, {=}\,x$, then $\g{x}{i} \,{>}\,  \g{y}{j}$. 
\paragraph{Maximality.}
We say that a glider $\plaing$ is \emph{maximal} in a set of gliders $\G$ if no $\plaing'\,{\in}\,\G$ satisfies $\plaing'\,{>}\,\plaing$.
In this sense, \emph{maximality} aligns with $I$-\emph{persistence}: a glider that is not dominated by any other conflicting one behaves stably {under the neighborhoods that are derived by $\G$}, while non-maximal gliders may be suppressed in evolution or reconstruction.
\begin{example}[Maximality under a strict partial order]\label{exmple:layering} \normalfont
    Let $\G \,{=}\, \{ \glider{1}, \glider{2}, \glider{3}, \glider{4} \}$ be a glider set with the following ordering relations: 
    $\glider{1} \,{>}\, \glider{2}, \hspace{1mm}
    \glider{2} \,{>}\, \glider{3}, \hspace{1mm}
    \glider{1} \,{>}\, \glider{3}$ and $
    \glider{4} \,{>}\, \glider{3}.$
    Note that the strict partial order, $>$, allows glider sets to have more than one maximal element. For example, the maximal gliders of $\G$ are $\glider{1}$ and $\glider{4}$.
    \end{example}
\paragraph{Sound Derivation.}
We have defined the set of gliders $\G$ and the strict partial order $>$ derived from a CA $\mathcal{A}$ and an initial configuration $I$. 
If the derivation is sound, as formally defined next, it would be possible to reconstruct the CA from $(\G,>)$.
Formally, given $\mathcal{A}\,{\in}\,\caclassshort{1}{r}{\Sigma}$ and $I \,{\subseteq}\, C_{\mathcal{F}}$.
Let $\G$ and $>$ be the set of gliders and their order, derived from $\mathcal{A}|_I$ as defined above. 
We say that the \emph{derivation is sound} if the following conditions are satisfied.
\begin{enumerate}[i., topsep=1.5mm,itemsep=1.5mm]
    \item The graph induced by $>$ on $\G$ is acyclic,
    \item Every $w \,{\in}\, \FN_{\mathcal{A},I}$ has some $\g{\sigma}{i} \,{\in}\, \G$ with $w \,{\in}\, \LocAt{i}{\sigma}$ such that $f(w) \,{=}\,  \sigma$, and
    \item For every $w \,{\in}\, \FN_{\mathcal{A},I}$ if $\G_w \,{\subseteq}\, \G$ is the set of gliders $\{\g{\sigma}{i} \mid  w\in\LocAt{i}{\sigma}\}$, and $\g{\tau}{j}$ is maximal in $\G_w$ then $f(w)\,{=}\,\tau$.
\end{enumerate}
The first condition establishes that the dominance relation is well defined.
The second condition makes sure that every feasible neighborhood is defined according to a glider, and that this glider is in $\G$. 
Lastly, the third condition states that if the evolution of a certain feasible neighborhood $w$ intersects with more than one glider of $\G$ and can presumably be defined by any such glider, then the one that actually determines the evolution is the most dominant among these.

\smallskip

We now consider the inverse direction: Given a finite set of gliders $\G$ equipped with a partial order $>$, can we reconstruct a CA whose local rule is derived from $(\G,>)$? 
That is, does there exist a CA that this glider system \emph{induces}?

     \subsection{Glider-Based Rule Reconstruction}\label{sec:part-order} 
    Given a set of gliders $\G$ ordered by $>$, we construct a {(possibly partial) rule function consistent with their dominance structure.
 \begin{example}\label{exmp:motivation-glider}\normalfont
        Let $\G\,{=}\,\{\g{a}{0}, \g{b}{1}\}$ over $\Sigma \,{=}\allowbreak \{a,b\}$ with $r\,{=}\,1$, and assume $\g{b}{1} \,{>}\, \g{a}{0}$. 
        Then, for any $w\,{\in}\,\textstyle{\Sigma^{3}}$, the rule induced by $\G$ satisfies: if $\sig{{-1}}\,{=}\, b$, then $f(w)\,{=}\,b$; otherwise, if $\sig{0}\,{=}\, a$, then $f(w)\,{=}\,a$.
    \end{example}
  We briefly describe the procedure, $\mainproc$, which derives a partial rule function $f$ from an ordered glider system $(\G,>)$.
  Each glider $\g{\sigma}{i} \,{\in}\,\G$ constrains $f$ to map certain neighborhoods in $\LocAt{i}{\sigma}$ to $\sigma$, unless a more dominant glider overrides it.
  The procedure utilizes simple subroutines that process gliders in dominance order and incrementally assign values to $f$.
   
\paragraph{Maximal Glider ($\procmaxx$).} 
The procedure ${\procmaxx(\G,>)}$ returns a pair $(\sigma,i)$ such that $\g{\sigma}{i}\,{\in}\,\G$ is a \emph{maximal} glider in $\G$ with respect to $>$.

\paragraph{Rule Extension ($\extend$).} 
Given a value $\sigma$, a position $i$, and a partial rule function $f'\,{:}\,\SigmaWindow{\rightharpoonup} \Sigma$, the $\extend$ procedure extends $f'$ by assigning $\sigma$ to every neighborhood $w\in \LocAt{i}{\sigma}$ for which $f'(w)$ is undefined. 
Formally, we define $\boldsymbol{f'} \,{:=}\, \extend(f', \sigma,i)$, where:
\vspace{-2.5mm}

\begin{equation*}
    \boldsymbol{f'}(w)\,{:=}\,\begin{cases}
    f'(w) & \text{if } w\,{\in}\,\dom(f')\\
     \sigma & \text{if } w \,{\in}\,\LocAt{i}{\sigma} \,{\setminus}\, \dom(f')
\end{cases}
\end{equation*}

\vspace{-1mm}

\subsection*{Glider-to-Rule ($\mainproc$) Reconstruction Algorithm}
\begin{wrapfigure}{r}{0.445\textwidth}
\vspace{-15mm}
\centering
\begin{minipage}{0.435\columnwidth}
\begin{algorithm}[H]
\small
\caption{{\mainproc}{($\G,>,\Sigma,r$)}}\label{alg:f-by-G}
\begin{algorithmic}[1]
\State $f' \gets \emptyset$ 
\While{$\G  \,{\neq}\,  \emptyset$}
    \State $(\sigma,i) \gets$ \Call{$\procmaxx$}{$\G$} 
    \State $\G \gets \G\setminus \{\g{\sigma}{i}\}$  
    \State $f' \gets$\Call{Extends}{$f',\sigma,i$}
\EndWhile
\State \Return $f'$  
\end{algorithmic}
\end{algorithm}
\end{minipage}
\vspace{-2mm}
\end{wrapfigure}
We define $\mainproc(\G, >, \Sigma, r)$ to construct a partial rule $f' {:}\SigmaWindow { \rightharpoonup} \Sigma$ by iteratively applying $\extend$ to maximal gliders, following any linearization of $>$, until all of $\G$ is processed.\footnote{A linearization of a strict partial order is a total order that respects the partial order: if $\gglider {>} \gglider'$, then $\gglider$ precedes $\gglider'$ in the sequence.}

We now show that if a CA derives a sound glider system $(\G,>)$, then its language can be reconstructed from this set.

\begin{lemma}[Soundness of Glider Reconstruction]\label{lem:glider-sound}
Let $\mathcal{A} \,{\in}\, \caclassshort{1}{r}{\Sigma}$ and $I \,{\subseteq}\, C_{\mathcal{F}}$. Suppose that $(\G,>)$ is the glider system \emph{derived} from $\mathcal{A}|_{I}$.
If the derivation is sound, then the partial rule $f' \,{=}\, \mainproc(\G, >, \Sigma, r)$ satisfies $f'(w) \,{=}\, f(w)$ for all $w \,{\in}\, \FN_{\mathcal{A},I}$.
\end{lemma}

\begin{proof}
    Let $(\G,>)$ be the glider system \emph{derived} from $\mathcal{A}|_I$.
    Assume towards contradiction that although the derivation is sound, there exists a $w \,{\in}\, \FN_{\mathcal{A},I}$ for which the partial rule $f' \,{=}\, \mainproc(\G, >, \Sigma, r)$ returns $f'(w) \,{\neq}\, f(w)$.
    By $\mainproc$ construction, the only gliders that can define the assignment of $w$'s rule are from $\{\g{\sigma}{i} ~|~ w\in\LocAt{i}{\sigma}\}$, denoted $\G_w$ in the sound-derivation definition. Procedure $\mainproc$ chooses a maximal glider among them. By item (iii) of the sound derivation, this maximal glider, $\g{\tau}{j}$, imposes $f(w)\,{=}\,\tau$.
    Thus, $f'(w)$ would also assign it $\tau$, which contradicts $f'(w)\,{\neq}\,f(w)$.
    Therefore, for every sound derivation, the partial rule function satisfies $f'(w)\,{=}\,f(w)$ for every $w\,{\in}\, \FN_{\mathcal{A},I}$.
    \qed
\end{proof}
\begin{remark}
Note that in general, there are multiple linearizations that $\mainproc$ could follow and that different linearizations, $l'$ and $l''$, yield different rule functions, $f'$ and $f''$ resp. 
However, given that $(\G,>)$ was derived from $\mathcal{A}\in\caclassshort{1}{r}{\Sigma}$ and $I$ as per Sec.\ref{sec:derived}, both will agree with $f$ on any feasible neighborhood with respect to $\mathcal{A}\,|_I$. That is, $f(w)=f'(w)=f''(w)$ for every $w\,{\in}\, \FN_{\mathcal{A},I}$.
\end{remark}
\begin{definition}[Glider Expressible]\label{def:glider-exp}
    Let $\mathcal{A} \,{\in}\, \caclassshort{1}{r}{\Sigma}$ and $I \,{\subseteq}\, C_{\mathcal{F}}$. We term language $\mathcal{L}(\mathcal{A},I)$ \emph{glider-expressible} if the glider system derived from $\mathcal{A}|_I$ is sound.
\end{definition}
We denote the class of such languages by $\caclassglider{\Sigma}{r}$. 
The following inclusion is a direct consequence of Lem.\ref{lem:glider-sound}.
\begin{corollary}
\label{cor:glider-implies-ca} 
    $\caclassglider{\Sigma}{r}\subseteq \caclassshort{1}{r}{\Sigma}$
\end{corollary}
That is, every glider-expressible language is also CA-expressible. However, the converse does not hold in general; there exist CA-expressible languages that are not glider-expressible. An explicit counterexample is provided in Ex.\ref{non-glider-expressible}. 
It follows from the following useful observation.

\begin{note}\label{obs:non-glider}
    Let $\mathcal{A}\,{\in}\,\caclassshort{1}{r}{\Sigma}$, $I\,{\subseteq}\, C_{\mathcal{F}}$,
    and $L\,{=}\,\mathcal{L}(\mathcal{A},I)$. If there exists a $w\,{\in}\,\FN_{\mathcal{A},I}$ such that $f(w){=}\,\sigma$ while $\sigma\,{\notin} \{\sig{i} \,|\,i\,{\in}\, [{-}r,r] \}$, then $L\notin \caclassglider{\Sigma}{r}$.
\end{note}
This phenomenon occurs in the following example.
\begin{example}[Rule 90 \cite{Wolfram2002}]\label{non-glider-expressible} \normalfont
        Let $\mathcal{A}\,{\in}\,\caclassshort{1}{1}{\{0,1\}}$ where $f(xyz) \,{=}\,  x \,{\oplus}\, z$, for every $x,y,z\,{\in}\, \{0,1\}$ and $\oplus$ is the XOR operation. %
        Note that the assignment $f(111) \,{{=}}\,  0$, is not derivable from any glider $\g{\sigma}{i}$ with $i\,{\in}\,\{-1,0,1\}$.
        Hence, this CA (aka Wolfram's Elementary-CA \emph{Rule 90}) is not glider-expressible.
 \end{example}
\begin{corollary}
\label{glidre-exp-implies-ca-exp}
    $\caclassglider{\Sigma}{r} \subsetneq \caclassshort{1}{r}{\Sigma}$.
    \end{corollary}
   \section{$\caclassglider{\Sigma}{r}$ Can Count}\label{subsec:expressivness}
   We turn to show that $\caclassglider{\Sigma}{r}$ (therefore $ \caclassshort{1}{r}{\Sigma}$) can generate languages beyond the OCA hierarchy. 
    A classic example is the language $\{a^nb^nc^n\,|\,n\,{\in}\,\mathbb{N}\}$, which is known to be context-sensitive but not context-free, whereas all OCA-recognizable languages are context-free. 
    Fig.\ref{fig:OCA-lang} presents a local rule function that generates this language for $I\,{=}\,\ppad(\{abc\})$, along with an illustration of its evolution.
    In this section, we explore a powerful generalization of this result, given in Thm.\ref{Thm:glider-expr-pwr}. We progress towards it gradually, generalizing one aspect at a time.

Our first result shows that CA can count any given number of letters with different multiplications in their powers.

   \begin{proposition}\label{prop:421}
        Let $L \,{=}\, \{ {\apos{1}}^{k_1{\cdot}n} {\apos{2}}^{k_2{\cdot}n}{\ldots} {\apos{m}}^{k_m{\cdot}n} \mid n\,{\in}\,\mathbb{N}\}$ where ${\apos{i}} \,{\in}\, \Sigma$ and \allowbreak{$k_i\,{\in}\, \mathbb{N}$} for every $i\,{\in}\, [m]$.
        If $\sum_{i {=}1}^m {k_i} \leq 2r$, then $L$ is {$\caclassshortREG{1}{r}{\Sigma}{REG}$}-expressible.
    \end{proposition}
    We can assume without loss of generalization that $m \,{=}\,  2r$, by padding $\mybot$ if needed and rewriting $a_i^{k_i\cdot n}$ as $a_i^n\cdots a_i^n$ (with $k_i$ occurrences). By relabeling the first $r$ symbols as $\aneg{r},{\ldots},\aneg{1}$ and the last $r$ as $\apos{1},{\ldots} \apos{r}$, the result becomes:
     \begin{restatable}[Prop.\ref{prop:421} reformulated]{proposition}{propnonk}\label{prop:421- nonk}
        Let $r\,{\in}\,\mathbb{N}$ and $L_r$ be the language $\{ \aneg{r}^{n} {\ldots} \aneg{1}^{n} \apos{1}^{n} {\ldots} \apos{r}^{n}\allowbreak \,|\, n\,{\in}\,\mathbb{N}\}$ over $\Sigma \,{=}\, \{\mybot\}\cup \{{\aneg{i}}, {\apos{i}}\,|\, i\,{\in}\,[r]\}$.
        Then $L_r$ is \emph{glider-expressible}, that is, $L_r\,{\in}\,\caclassglider{\Sigma}{r}$, and thus is also {$\caclassshortREG{1}{r}{\Sigma}{REG}$}-expressible.
    \end{restatable}
    Before providing the proof, we illustrate the construction in the concrete case of $r\,{=}\,2$ in Ex.\ref{exmp:anbncndn}, where the symbols $\aneg{2},\aneg{1},\apos{1},\apos{2}$ are renamed $a,b,c,d$ respectively for readability. Note that the base case of $r \,{=}\,1$ is given in Clm.\ref{claim:CAvocl}.
    
       \begin{example}\label{exmp:anbncndn}\normalfont 
Let $r \,{=}\, 2$, $L \,{=}\, \{ a^n b^n c^n d^n ~|~ n \,{\in}\, \mathbb{N} \}$, and $F \,{=}\, \{ abcd \}$.
This corresponds to $L_2$ with $\aneg{2} \,{=}\, a$, $\aneg{1} \,{=}\, b$, $\apos{1} \,{=}\, c$, $\apos{2} \,{=}\, d$. The CA uses eight gliders, two for each symbol; they propagate the respective symbol at the appropriate velocities.
The inner letters ($b$ and $c$) need to shift less than the outer letters ($a$ and $d$). 
As a result, $b$ and $c$ override locations previously occupied by $a$ and $d$ resp. 
Formally, lower absolute index gliders ($|i|$, for $\apos{i}$), dominate the higher ones when overlapping, ensuring nested growth. App.\ref{subsec:rule-usage} provides a detailed illustration of $f$'s behavior.
\end{example}

\vspace{-2.25mm}

\begin{proof}[sketch]
We construct a CA $\mathcal{A} \,{\in}\, \caclassshortREG{1}{r}{\Sigma}{REG}$ that generates $L_r$ from an initial configuration $I\,{=}\,\ppad (\{\aneg{r}{\ldots}\aneg{1}\apos{1}{\ldots}\apos{r}\})$. 
Each symbol $i\,{\in}\,[-r,r]{\setminus}\{0\}$ is assigned to two gliders $\G_i\,{=}\,\{\g{\apos{i}}{i},\g{\apos{i}}{i-\sign(i)}\}$ with a velocity relative to $i$.
The set of induced gliders $\G \,{=}\, \bigcup_{i\in [r]} \G_i \cup \G_{-i}$, is associated with a strict partial order $>$, where; $\g{\apos{i}}{x} > \g{\apos{j}}{y}$ iff $\sign(i) \,{=}\, \sign(j)$, $|i|\,{<}\,|j|$, $\apos{i}\,{\neq}\,\apos{j}$, and $x\,{\neq}\, y$. As gliders with the same symbol or the same velocity are incomparable under $>$. 
Accordingly, the maximal gliders are those of the symbols $\aneg{1}$ and $\apos{1}$.
That is, the gliders of $\G_1$ and $\G_{-1}$ are maximal in $\G$; other gliders are suppressed when dominated.
The local rule function of the CA induced by applying $\mainproc$ to ${(\G,>)}$, determines the assignment of a neighborhood respecting the dominance relation induced by the gliders.   
This ensures deterministic growth of the nested pattern $c_n \,{=}\, {}^\omega \mybot\, \aneg{r}^{n+1} {\ldots} \aneg{1}^{n+1} \apos{1}^{n+1} {\ldots} \apos{r}^{n+1} \,\mybot^\omega$,
such that $\mathcal{L}(\mathcal{A},I) \,{=}\, L_r$. 
The full proof appears in App.\ref{app:full-proof-4.2.1}. \qed
\end{proof}
     \begin{figure}[ht]
     \centering{
    \refstepcounter{figure}
\addtocounter{figure}{-1}
        \scalebox{0.935}{{$\begin{array}{|l|l|}
        \hline
        \mybot  & \mybot \mybot \mybot \mybot \mybot   \hspace{0.525em}  \mybot \mybot \mybot \mybot a \\\hline
        a & \mybot \mybot \mybot aa  \hspace{0.525em}  \mybot \mybot \mybot ab  \hspace{0.525em} \mybot \mybot aaa  \hspace{0.525em} \mybot \mybot aab  \hspace{0.525em} \mybot \mybot abc  \hspace{0.525em} \mybot aaaa  
        \hspace{0.525em}
         \mybot aaab   \\
        & \text{\makebox[\widthof{$\mybot \mybot \mybot aa$}][c]{$\mybot aabb$}} \hspace{0.525em}  \text{\makebox[\widthof{$\mybot \mybot \mybot ab$}][c]{$aaaaa$}}  \hspace{0.525em} \text{\makebox[\widthof{$\mybot \mybot aaa$}][c]{$aaaab$}}  \hspace{0.525em}  \text{\makebox[\widthof{$\mybot \mybot aab$}][c]{$aaabb$}} \\\hline
        b & \text{\makebox[\widthof{$\mybot \mybot \mybot aa$}][c]{$\mybot abc\mybot$}}   \hspace{0.525em}  \text{\makebox[\widthof{$\mybot \mybot \mybot ab$}][c]{$aabbb$}}  \hspace{0.525em} \text{\makebox[\widthof{$\mybot \mybot aaa$}][c]{$aabbc$}}  \hspace{0.525em} \text{\makebox[\widthof{$\mybot \mybot aab$}][c]{$abbbb$}}  \hspace{0.525em} \text{\makebox[\widthof{$\mybot \mybot abc$}][c]{$abbbc$}}  \hspace{0.525em} \text{\makebox[\widthof{$\mybot aaaa$}][c]{$abbcc$}}  \hspace{0.525em} abc\mybot \mybot  \\
        & \text{\makebox[\widthof{$\mybot \mybot \mybot aa$}][c]{$bbbbb$}} \hspace{0.525em} \text{\makebox[\widthof{$\mybot \mybot \mybot ab$}][c]{$bbbbc$}}  \hspace{0.525em} \text{\makebox[\widthof{$\mybot \mybot aaa$}][c]{$bbbcc$}}  \hspace{0.525em} \text{\makebox[\widthof{$\mybot \mybot aab$}][c]{$bbcc\mybot$}}   \hspace{0.525em} \text{\makebox[\widthof{$\mybot \mybot abc$}][c]{$ bbccc$}} \\\hline
        c & \text{\makebox[\widthof{$\mybot \mybot \mybot aa$}][c]{$bc\mybot \mybot \mybot$}}   \hspace{0.525em} \text{\makebox[\widthof{$\mybot \mybot \mybot ab$}][c]{$bcc\mybot \mybot$}}    \hspace{0.525em} \text{\makebox[\widthof{$\mybot \mybot aaa$}][c]{$bccc\mybot$}}   \hspace{0.525em} \text{\makebox[\widthof{$\mybot \mybot aab$}][c]{$bcccc$}}  \hspace{0.525em} c{{\bot} \bot} {\bot \bot}   \hspace{0.525em} cc\mybot \mybot \mybot  
        \hspace{0.525em}
        ccc\mybot \mybot  \\
        & \text{\makebox[\widthof{$\mybot \mybot \mybot aa$}][c]{$cccc\mybot $}}   \hspace{0.525em} \text{\makebox[\widthof{$\mybot \mybot \mybot ab$}][c]{$ccccc$}} \\
        \hline
        \end{array}$}}

  \vspace{3mm}

         \scalebox{0.365}{
     \begin{tikzpicture}[baseline={(current bounding box.center)}]
     \draw[step=1cm, gray, thin] (7, 1) grid (20, -3);
        \foreach \i in {1,2,3,4} {
            \node at (6.25, 1.5-\i) {\huge ${\ldots}$};
            \node at (20.75,  1.5-\i) {\huge ${\ldots}$};
        }
        \fill[lightgray!50] (12, 0) rectangle ++(-1, 1);
        \fill[darkgray!50] (12, 0) rectangle ++(1, 1);
        \fillc{13}{0}{1}

        \fill[lightgray!50] (12, -1) rectangle ++(-2, 1);
        \fill[darkgray!50] (12, -1) rectangle ++(2, 1);
        \fillc{14}{-1}{2}
        
        \fill[lightgray!50] (12, -2) rectangle ++(-3, 1);
        \fill[darkgray!50] (12, -2) rectangle ++(3, 1);
        \fillc{15}{-2}{3}

        \fill[lightgray!50] (12, -3) rectangle ++(-4, 1);
        \fill[darkgray!50] (12, -3) rectangle ++(4, 1);
        \fillc{16}{-3}{4}

       \draw[step=1cm, black, thin] (7, 1) grid (20, -3);
       \draw[step=1cm, black, ultra thick, ->] (5.25,1) -> (5.25,-3); 
        \node[rotate=90] at (4.75, -1) {\Huge {Time}};
        \node[] at (13.5,-3.7) {\Huge \textbf{Illustration of}~$I\,{=}\, {}^{\omega}\mybot abc \mybot^{\omega}$ };
        \end{tikzpicture}
    }
}

 \caption[Example for $\{a^n b^n c^n \,|\, n\,{\in}\,\mathbb{N}\}$]{%
  $\{a^n b^n c^n \,|\, n\,{\in}\,\mathbb{N}\}$ {Legend:} $\bot$ {=} $\whitebox$, $a$ {=} $\abox$, $b$ {=} $\bbox$, $c$ {=} $\cbox$)%
}\label{fig:OCA-lang}
\end{figure}
The construction above demonstrates that CA which are \emph{glider-expressible}, can produce arbitrarily segment-aligned languages with precise symbolic spacing and coordination.
Crucially, the mechanism used to generate $L_r$ (i.e., the induced $(\G,>)$), does not depend on the specific value of each symbol.
Rather, it depends on the \emph{structure} and positioning of the segments, as well as on their bounded total width.
We now generalize this idea to arbitrary repeated blocks. 
   
\begin{restatable}{theorem}{gliderExp}\label{thm:w1ntowmn-glider}
        Let $m\,{\in}\,\mathbb{N}$, $w_1,{\ldots},w_m\,{\in}\,{\Sigma^{+}}$ and $r\,{=}\,\lceil\frac{1}{2}\sum_{i{=}1}^m |w_i|\rceil$.
        If the language $L\,{=}\,\{w_1^n w_2^n{\ldots} w_m^n~|~n\,{\in}\,\mathbb{N}\}$ is in $ \caclassshort{1}{r}{\Sigma}$, then, in particular, $L\,{\in}\, \caclassglider{\Sigma}{r}$.
    \end{restatable}

\noindent
To prove this, we build on the idea that gliders can encode both \emph{repetition} and \emph{shifting}.
Intuitively, the construction uses gliders to repeatedly emit new copies of each $w_i$, in the correct order and alignment, from a compact initial configuration.
The following lemma provides the base case and will be used as a building block.

\begin{restatable}[Shift {\&} Concat]{lemma}{shiftlemm}\label{lemma:shift}
For any $w \,{\in}\, {\Sigma^{+}}$ and direction $d \,{\in}\, \{\emph{left}, \emph{right}\}$, there exists $\mathcal{A} \,{\in}\,\caclassglider{\Sigma}{r}$ that repeatedly shifts the existing copies of $w$ by $s\,{\in}\,\mathbb{Z}$ and appends a new copy to the $d$ side, resulting in the language $\{w^n \,|\, n\,{\in}\,\mathbb{N}\}$.
\end{restatable}

\noindent
The idea is to use persistent gliders with aligned shift velocities to {replicate} $w$ incrementally. Fig.\ref{fig:small-shift-concat} illustrates a word $w\,{=}\,\bbox\abox\abox\cbox$ that is shifted one step to the right, (i.e., $s\,{=}\,1$) and is concatenated to the \emph{right} of the shifted block.
This results in the glider set $\G\,{=}\,\left\{\g{\sigma}{+1}, \, \g{\sigma}{+5} ~|~ \sigma\,{\in}\,\{\abox,\bbox,\cbox\}\right\}$. To avoid clutter, we only annotated the gliders of $\bbox$ and $\cbox$ in the first step.
The proof of the lemma is given in App.\ref{app:shift-gliders}.
 \begin{figure}[ht]
    \centering
    \vspace{1mm}
    {\small 
     \begin{tikzpicture}[scale=0.445, every node/.style={scale=0.445}]
        \foreach \i in {1,2,3} {
            \node at (5.5, 1.5-\i) {\huge ${\ldots}$};
            \node at (22.5,  1.5-\i) {\huge ${\ldots}$};
        }
        \fill[darkgray!50] (7, 0) rectangle ++(1, 1);
        \fill[lightgray!50] (8, 0) rectangle ++(2, 1);
        \fillc{10}{0}{1}
        \fill[darkgray!50] (8, -1) rectangle ++(1, 1);
        \fill[lightgray!50] (9, -1) rectangle ++(2, 1);
        \fillc{11}{-1}{1}
        \fill[darkgray!50] (12, -1) rectangle ++(1, 1);
        \fill[lightgray!50] (13, -1) rectangle ++(2, 1);
        \fillc{15}{-1}{1}
        \fill[darkgray!50] (9, -2) rectangle ++(1, 1);
        \fill[lightgray!50] (10, -2) rectangle ++(2, 1);
        \fillc{12}{-2}{1}
        \fill[darkgray!50] (13, -2) rectangle ++(1, 1);
        \fill[lightgray!50] (14, -2) rectangle ++(2, 1);
        \fillc{16}{-2}{1}
        \fill[darkgray!50] (17, -2) rectangle ++(1, 1);
        \fill[lightgray!50] (18, -2) rectangle ++(2, 1);
        \fillc{20}{-2}{1}
        \draw[step=1cm, darkgray, thin] (6, 1) grid (22, -2);
        \blackonwhiteArrow{7.5}{0.5}{12.5}{-0.5}{thick}{black}
         \blackonwhiteArrow{7.5}{0.5}{8.5}{-0.5}{thick}{darkgray}
        \blackonwhiteArrow{10.5}{0.5}{15.5}{-0.5}{thick}{black}
        \blackonwhiteArrow{10.5}{0.5}{11.5}{-0.5}{thick}{darkgray}
        \blackonwhiteScaledOne{8.25}{-0.95}{\gfig{\smallbbox}{\boldsymbol{+1}}}{\huge}{\gfig{\smallwwhitebox}{\boldsymbol{+1}}}
        \blackonwhiteScaledOne{11.25}{-0.95}{\gfig{\smallcbox}{\boldsymbol{+1}}}{\huge}{\gfig{\smallwwhitebox}{\boldsymbol{+1}}}
        \blackonwhiteScaledOne{16.55}{-0.95}{\gfig{\smallcbox}{\boldsymbol{+5}}}{\huge}{\gfig{\smallwwhitebox}{\boldsymbol{+5}}}
        \blackonwhiteScaledOne{13.55}{-0.95}{\gfig{\smallbbox}{\boldsymbol{+5}}}{\huge}{\gfig{\smallwwhitebox}{\boldsymbol{+5}}}
          \draw[step=1cm, black, semithick, ->] (4.85,1) -> (4.85,-2); 
   
        \node[rotate=90] at (4.25, -0.65) {\huge {Time}};
        \end{tikzpicture}
    }
    \caption[]{%
  $w\,{=}\,\bbox\abox\abox\cbox$, $s\,{=}\,\text{+}1$, $d\,{=}\,$right
}\label{fig:small-shift-concat}
 \end{figure}
We now extend to two block patterns $L\,{=}\,\{w_1^n w_2^n \,|\, n\,{\in}\,\mathbb{N}\}$, which introduce asymmetry and spatial coordination.
When, without loss of generalization $|w_1|\,{>}\,|w_2|$, naive repetition breaks alignment; we overcome this by letting the boundary between the blocks ``move dynamically'' during growth.

\vspace{-1.25mm}
\begin{restatable}[Two-block case]{lemma}{wnwn}\label{lemma-lem2}
        Let $w_1, w_2\,{\in}\, \Sigma^+$, and define $r\,{=}\, \lceil{\frac{|w_1|+|w_2|}{2}}\rceil$.
        Then $L\,{=}\,\{w_1^n w_2^n \,|\allowbreak\, n\,{\in}\, \mathbb{N}\}\,{\in}\,\caclassglider{\Sigma}{r}$.
\end{restatable}  
\noindent
The gliders that are associated with $w_1$ and $w_2$ interleave and overwrite part of each other's space, maintaining consistent symbolic alignment through carefully assigned shifts and dominance. See App.\ref{app:two-blocks} for the full proof.

In the general case, for arbitrary many words of arbitrary lengths, no global shift suffices. Each $w_i$ is defined with its own \emph{shift} and glider family. 
A key insight is that if $L$ is $\caclassshort{1}{r}{\Sigma}$-expressible, then its derived gliders and partial order are sound.
Thus, the CA can be reconstructed from the gliders.
Moreover, this construction naturally extends to languages where the powers differ per block. Formally,

\begin{restatable}{theorem}{corGliderExp}\label{Thm:glider-expr-pwr}
        Let $m\,{\in}\,\mathbb{N}$, $w_1,{\ldots}, w_m\, {\in}\, {\Sigma^+}\!$, $r \,{=}\, \lceil\frac{1}{2}\sum_{i{=}1}^m \, a_i{\cdot}|w_i|\rceil$ and $e_i(n)\!=\allowbreak a_i\,{\cdot}\,n\,{+}\,b_i$ be positive linear expressions, i.e., $a_i,\,b_i\,{\in}\,\mathbb{N}_0$ and $a_i\,{>}\,0$ for any $i\,{\in}\,[m]$.
        Consider $L = \{w_1^{e_1(n)}{\ldots} w_m^{e_m(n)}|\, n\,{\in}\,\mathbb{N}\}$.
        If $L\,{\in}\,  \caclassshort{1}{r}{\Sigma}$, then $L \,{\in}\,\caclassglider{\Sigma}{r}$.
\end{restatable}

\noindent
The proof of Thm.\ref{Thm:glider-expr-pwr}, including glider resolution, appears in App.\ref{app:block-repetition}.
A specific example, $(aba)^{2n}c^{n+1}$, is given in App.\ref{full-complex}.
    
\section{Discussion}\label{sec:discussion}
We set out to explore the symbolic expressiveness of deterministic, one dimensional cellular automata (CA) through a new lens--the \emph{generative perspective}.
To this end, we introduced a glider-based framework, where global behavior is induced by finitely many local symbolic patterns (``gliders'') equipped with a strict partial order that specifies their dominance relation upon collision. 
This dominance order formalizes which glider prevails in each interaction, ensuring that the resulting dynamics are compositional and well-defined. 
The framework naturally led us to define the class of \emph{glider-expressible languages}--those CA languages generated by sound glider systems associated with a CA in $\caclassshort{1}{r}{\Sigma}$, i.e., finitely many gliders interacting under such a dominance relation.

We showed that, despite their syntactic constraints, the class of \emph{glider expressible CA}, and therefore CA expressible in general, captures complex and interesting behaviors examples include: synchronized multi-counter patterns such as $a^n b^n c^n d^n$, which illustrate fair division of unbounded resources, and asymmetric growth patterns such as $(aba)^{2n}c^{n+1}$, which reflect skewed resource consumption.
These examples highlights that the glider generative model retains sufficient structural richness to encode nontrivial, long-range dependencies.

From a verification standpoint, this generative perspective provides a direct language-theoretic handle. 
Let $(\G,>)$ be a sound glider system derived from $\mathcal{A}\,{\in}\,\caclassshort{1}{r}{\Sigma}$ and an initial regular set $I \,{=}\, \ppad(L) \,{\subseteq}\,C_{\mathcal{F}}$, for $L\,{\in}\, \class{REG}$, and let $B\,{=}\,\ppad(\mathcal{B})$ denote the regular specification of bad configurations, where $\mathcal{B}\,{\in}\,\class{REG}$.
The safety question ``$\exists t \,{:}\, G^t(c_0)\,{\in}\, B$ for some $c_0\,{\in}\, I$?''
Can be reduced to checking whether $\mathcal{L}(\mathcal{A},I)\cap \mathcal{B} \neq \emptyset$. If the generated language is regular, classical forward or backward regular model checking applies directly~\cite{Bouajjani2000}. 
When $\mathcal{L}(\mathcal{A},I)$ is a \emph{context-free language} (CFL), the problem is still decidable~\cite{FismanP01}: CFLs are closed under intersection with regular languages, and the emptiness problem is decidable, enabling workflows in the spirit of forward regular model checking.
However, for systems whose generated language is not context-free (as in our multi-counter examples), this reduction no longer applies.
This motivates the need for new verification techniques capable of exploiting the structured generative semantics that our framework exposes.

Overall, the glider-based viewpoint provides an interpretable and compositional bridge between symbolic dynamics and language-theoretic verification. By exposing the generative mechanisms behind CA evolution, it supports reasoning about system that are simultaneously unbounded and highly structured, and highlights the potential of CA under this generative perspective; as analyzable models for emergent computations.

\begin{credits}
\subsubsection{\ackname}
Noa Izsak carried out this work in part as a member of the Saarbr\"ucken Graduate School of Computer Science.
\end{credits}

\bibliographystyle{splncs04}
\bibliography{references}

\newpage
\appendix
\begin{center}
    \Large \textbf{Appendix}
    \vspace{-1.5mm}
\end{center}
\section{Complete proofs - Properties of CA 
\Cref{subsec:CAprop}
}\label{app:CAproperties}
\vspace{-1mm}
\bounddiff*
\begin{proof}
         Let $\mathcal{A}\,{\in}\, \caclassshort{1}{r}{\Sigma}$, and let $c\,{\in}\, C_{\mathcal{F}}$ be a finite configuration. 
         By Cor.\ref{cor:bound_inc}, the gap between $\width{c}$ to $\width{G(c)}$ is $k$, which is at most $2{\cdot}r$.
         Let $k$ be the maximal value, that is, $2r$. 
         By the definition of \term{finite configuration}, only cells outside the quiescent neighborhoods can become active within $\interval{G(c)}$. 
         As $G(C)$ may affect only the boundary of the $\intervalInt$, precisely $2r$ such cells can appear. 
         \qed
         \end{proof}
Based on the above, another corollary emerges.
\begin{corollary}
        Let $L \,{=}\, \mathcal{L}(\mathcal{A},c)$ a $k$-difference bounded language. Then $L$ is not CA-expressible with $r\,{\leq}\,\lfloor\frac{k \,{-} \,1}{2}\rfloor$.
    \end{corollary}
    \begin{proof}%
        By the definition of \emph{$k$-difference-bounded language}, any two consecutive length words, i.e., any $w_j,w_{j+1}\,{\in}\, L$, satisfies $\width{w_{j+1}} {-}\width{w_j}\,{\leq}\, k$. 
        In particular, there exists an $i\,{\in}\,\mathbb{N}$ such that $\width{w_{i+1}} {-}\width{{w_{i}}}\,{=}\,k$ and there is no word of length between $\width{{w_{i}}}$ and $\width{w_{i+1}}$.
        Assume toward contradiction that there is a CA $\mathcal{A}\,{\in}\,\caclassshort{1}{r}{\Sigma}$ for $r\,{\leq}\,{\lfloor{\frac{k\mminus 1}{2}}\rfloor}$ such that $L$ is CA-expressible via $\mathcal{A}$.
        Without loss of generalization, assume that $^{\omega} \mybot w_{i+1} \mybot ^{\omega}  \,{\in}\, \orbitc{{}^{\omega} \mybot {w_{i}} \mybot ^{\omega} }$, in particular, that ${}^{\omega} \mybot w_{i+1} \mybot ^{\omega}\,{=}\,G({}^{\omega} \mybot {w_{i}} \mybot ^{\omega} )$. 
 
        Since, if ${}^{\omega} \mybot w_{i+1} \mybot ^{\omega} \,{=}\,G^n({}^{\omega} \mybot {w_{i}} \mybot ^{\omega} )$ for some $n\,{>}\,1$, then necessarily, for any ${w''} \,{=}\,  G^{n'}({}^{\omega} \mybot {w_{i}} \mybot ^{\omega} )$ with $n'\,{<}\,n$, it holds that $|\width{w''} {-}\width{w_{i}}| \,{\leq}\,k$ and similarly, $|\width{w''} {-}\width{w_{i+1}}| \,{\leq}\,k$.
        Thus, we assume that $w_{i+1}$ immediately follows ${w_{i}}$.
        By applying the global function:
        \begin{center}$\width{w_{i+1}}{-\,}\width{{w_{i}}}\,{=}\,\width{G({w_{i}})}{-\,}\width{{w_{i}}}$
        \end{center}
        By Cor.\ref{CApump:lemma} we have: $\width{G({w_{i}})}{-\,}\width{{w_{i}}}\,{\leq}\,2{\cdot}r$. 
        Assuming $r\,{\leq}\,\lfloor\frac{k\mminus 1}{2}\rfloor$, then $2{\cdot} r\,{\leq}\,k\,{-}\,1$.
        So we got that; $\width{w_{i+1}}{- \, } \width{w_{i}} \,{\leq}\,k{-}1 $, while we started by saying that $\width{w_{i+1}}{-}\, \width{w_{i}} \,{=}\, k$, \emph{contradiction}. Therefore, $L$ is not $\caclassshort{1}{r}{\Sigma}$-expressible for $r \,{\leq}\,\lfloor\frac{k\mminus 1}{2}\rfloor$.
        \qed
    \end{proof}
        \section{OCA Definitions -- \Cref{subsec:onecounter}
        }\label{app:oca-defs}
        \vspace{-1mm}
        \begin{definition}[One-counter automata (OCA)]
            An one-counter automaton $\mathcal{A}$ is a tuple $(\Sigma,Q,q_0,F,\delta)$ where:
         \begin{itemize}[topsep=0.25mm]
             \item $\Sigma$ is a finite input alphabet, 
             \item $Q$ is a finite set of states, with initial state $q_0\,{\in}\, Q$, 
             \item $F\,{\subseteq}\, Q$ is the set of final (accepting) states, and
            \item $\delta\,{:}\, Q\times (\Sigma \cup \{\varepsilon\} ) \times \{\,{=} 0,\ {>}0\} \rightharpoonup 2^{Q\times \{\miniplus 1, \miniminus 1, 0 \}}$ is the transition relation.
         \end{itemize}
        \end{definition}
    The OCA $\mathcal{A}$ is \term{sound} if for every $(q',op)\,{\in}\allowbreak\delta(q,\sigma,{{=} 0})$, where $q'$ is the new state and $op$ is the counter operation, i.e., $op\,{\in}\,\{{+}1, {-}1, 0 \}$; it holds that $op\,{\neq}\, {-} 1$, namely it never decrements if the counter is zero. We consider only \term{sound OCA}.
\vspace{-2mm}
\begin{definition}[Configuration]
A \term{configuration} of $\mathcal{A}$ is a tuple of $(q,n)\,{\in}\allowbreak Q\,{\times}\,\mathbb{N}$, where $q$ is the current state and $n$ is the current counter value. The initial configuration is $(q_0,0)$.
\end{definition}
\begin{definition}[Counted Run]
    A \term{counted run} of an OCA $\mathcal{A}$ is a sequence of configurations with a legal transition between any consecutive pair, annotated by the read symbols.
    That is, a counted run is either a single state $(p_0, n_0)$ or a non-empty sequence of transitions $(p_0,n_0){\xrightarrow{\sig{1}}} (p_1,n_1) {\ldots}{\xrightarrow{\sig{\ell}}} (p_\ell, n_\ell)$, abbreviated as $(p_0,n_0)\,{\xrightarrow{\sig{1} \ldots \sig{\ell}}}\, (p_\ell, n_\ell)$.
\end{definition}
\paragraph{Cruns.} The set of counted-runs of $\mathcal{A}$ is denoted $\crunso{\mathcal{A}}$.  

\vspace{-3mm}

\subsection*{One-counter language (OCL)}
    A language of an OCA $\mathcal{A}$, is the set of all words that label a counted run from $(q_0,0)$ to some configuration $(q,0)$ where $q\,{\in}\, F$.
\begin{definition}[OCA's Language]
        The language of an OCA $\mathcal{A}$, denoted $\mathcal{L(A)}$, is $\{w\,{\in}\, \Sigma^* \,|\, \exists\, q\,{\in}\, F \text{. } (q_0,0) {\xrightarrow{w}} (q,0)\,{\in}\, \crunso{\mathcal{A}}\}$.
        \end{definition}
\paragraph{Sub-classes of OCA.} 
While OCA are inherently non-deterministic, various works have considered their deterministic variant \cite{valiant1975deterministic,bohm2013equivalence,mathew2025learning}.
    An OCA is \term{deterministic} (\term{DOCA}) if (1) the right-hand side (RHS) of $\delta$ is either: a singleton or the empty set, and (2) 
    for all $q\,{\in}\, Q,\ \sigma\,{\in}\, \Sigma$ and $g\,{\in}\, \{\,{=}\,0,\ {>}0\}$ we have that $\delta(q,\sigma,g)$ is defined only when $\delta(q,\varepsilon,g)$ is not.
    That is, despite having $\varepsilon$-transitions, the automaton is considered deterministic, since an $\varepsilon$-transition from state $q$ with guard $g$, are allowed only if there is no $\sigma$-transition from $q$ with that guard $g$.
    Note, however, that a DOCA may have more than one run on a particular word since the number of $\varepsilon$-transitions may be unbounded.

    A DOCA is \term{real-time} (\term{ROCA}) if it does not have $\varepsilon$-transitions.
    Lastly, a ROCA is \term{visible} (\term{VOCA}) if it satisfies that for every $\sigma\,{\in}\,\Sigma$, 
    $q,p\,{\in}\, Q$,  $g\,{\in} \{{=} 0,\ {>}0\}$ and $o\,{\in} \{{{+}1},\allowbreak {{-}1},0\}$; if
    $\delta(q,\sigma,g)\,{\in}\, Q{\times} \{o\}$ then 
    $\delta(p,\sigma,g)\,{\in}\, Q{\times} \{ o\}$.
    That is, every time we read a symbol $\sigma \,{\in}\, \Sigma$, we must perform the same counter operation. 
    Therefore, the modification of the counter depends \emph{only} on the input symbol.\footnote{This is a restriction of \emph{visibly pushdown automata}~\cite{VisiblyPushdownLanguages04} to counter automata.}
    For example, the language $\{a^{n}b^{n}\,|\,n\,{\in}\,\mathbb{N}\}$ is recognized by a VOCA, whereas $\{a^{n}ba^{n}\,|\,n\,{\in}\,\mathbb{N}\}$ is not.
\vspace{-3mm}
     \subsection*{One counter languages and CA languages (Prop.\ref{OCA-proposition})}\label{app:OCLandCA}
The following languages can be used to differentiate the classes \cite{staquet2024active}. 
    \begin{enumerate}[itemindent=0.1cm,topsep=0.5mm,itemsep=1mm,parsep=1mm]
        \item\label{claim:vocl} $L_{VOCL} \,{=}\,  \{a^nb^n  \,|\,  n\,{>}\,0\} \,{\in}\, \class{VOCL} \setminus \class{REG}$ %
    
        \item\label{claim:rocl} $L_{ROCL} \,{=}\, \{a^n b a ^n  \,|\,  n\,{>}\,0\} \,{\in}\, \class{ROCL} \setminus \class{VOCL}$%
        
        \item\label{claim:docl} $L_{DOCL} \,{=}\,  \{ a^n b^m c  \,|\, n\,{\geq}\,m\,{>}\,0\} \,{\in}\, \class{DOCL} \setminus \class{ROCL}$ %
        
        \item\label{claim:ocl} $L_{OCL} \,{=}\,  \{ a^n b^m  \,|\,  n\,{\geq}\,m\,{>}\,0\}  \,{\in}\, \class{OCL} \setminus \class{DOCL}$%
    \end{enumerate}

 \smallskip
\noindent
    It is well established that visibly OCA (VOCA) are strictly less expressive than real-time OCA (ROCA), which are strictly less powerful than deterministic OCA (DOCA), all of which lie strictly within the full class of OCA.
    See \cite{Herbst1991,hopcroft_ullman_automata_2001,staquet2024active,VisiblyPushdownLanguages04} 
    for formal definitions and proofs. 
    Using the class of $\caclassshortREG{1}{r}{\Sigma}{REG}$, we show language interactions for each of these complexity levels at \Cref{subsec:onecounter}.
\section{Evolution of $L \,{=}\, \{ \aneg{2}^n\aneg{1}^n \apos{1}^n  {\apos{2}}^n \,|\, n\,{\in}\,\mathbb{N}\}$ (\Cref{exmp:anbncndn})}\label{subsec:rule-usage}     
\emph{Intuition.} Before we dive into the example, we wish to provide some intuition about the construction. Let us first recall what this CA must do:
\begin{itemize}[itemindent=0cm,topsep=0.5mm,itemsep=0.5mm,parsep=0mm]
    \item Given initial configuration of one instance of each $\aneg{2},\aneg{1},\apos{1}$, and $\apos{2}$, it must \textbf{grow} the string on both sides such that at time step $n$ (i.e., $G^{n-1}(c_0)$), the configuration contains: $\aneg{2}^n \aneg{1}^n \apos{1}^n \apos{2}^n$.
    \item The left growing components $\aneg{2},\aneg{1}$ must \emph{push} left, and right-growing components $\apos{1},\apos{2}$ must push right.
\end{itemize}
\paragraph{The challenge.} is to ensure that the four symbols grow in the correct counts and positions, maintaining strict block separation: $\aneg{2}^n\aneg{1}^n \apos{1}^n  {\apos{2}}^n$
\paragraph{Strategy Overview.} To grow this structure, the CA uses gliders that propagate and interact to generate new copies of the required symbols.
\begin{itemize}[itemindent=0cm,topsep=0.5mm,itemsep=0.5mm,parsep=0mm]
    \item The evolution works in \textbf{waves}: at each step, the CA expands the number of copies of each symbol by one while preserving their relative order and overall block structure.
    \item Essentially, each symbol \emph{copies itself} into new positions -- this copying may \emph{overwrite} symbols from other blocks, which then must also regenerate ``more aggressively'' to recover both their lost symbols and continue growing.
\end{itemize}
 \begin{figure}[hb]
    {\small
\begin{multicols}{4}
    \begin{enumerate}
    \setlength{\itemindent}{-1em}
  \setlength{\leftmargin}{1.5em}
  \setlength{\itemsep}{0.3em}
  \setlength{\labelwidth}{1.25em}
  \setlength{\labelsep}{0.25em}
  \renewcommand{\labelenumi}{{\scriptsize\theenumi.}} 
        \item $\text{\makebox[\widthof{$f(d \mybot \mybot \mybot \mybot )$}][c]{$f(ddddd)$}}\,{=}\,d$
        \item $\text{\makebox[\widthof{$f(d \mybot \mybot \mybot \mybot )$}][c]{$f(dddd \mybot)$}}\,{=}\,d$
        \item ${\text{\makebox[\widthof{$f(d \mybot \mybot \mybot \mybot )$}][c]{$f(ddd \mybot \mybot )$}}\,{=}\,d }$
        \item\label{item:d:dd---} ${\text{\makebox[\widthof{$f(d \mybot \mybot \mybot \mybot )$}][c]{$f(dd \mybot \mybot \mybot )$}}\,{=}\,d }$
        \item\label{item:d:d----}  ${f(d\mybot \mybot \mybot \mybot )\,{=}\,d }$
        \item ${\text{\makebox[\widthof{$f(d \mybot \mybot \mybot \mybot )$}][c]{$f(cdddd)$}}\,{=}\,d }$
        \item ${\text{\makebox[\widthof{$f(d \mybot \mybot \mybot \mybot )$}][c]{$f(cddd\mybot )$}}\,{=}\,d }$
        \item\label{item:d:cdd--} ${\text{\makebox[\widthof{$f(d \mybot \mybot \mybot \mybot )$}][c]{$f(cdd\mybot \mybot )$}}\,{=}\,d }$
        \item\label{item:d:cd---} ${\text{\makebox[\widthof{$f(d \mybot \mybot \mybot \mybot )$}][c]{$f(cd\mybot \mybot \mybot )$}}\,{=}\,d }$
        \item ${\text{\makebox[\widthof{$f(d \mybot \mybot \mybot \mybot )$}][c]{$f(ccddd)$}}\,{=}\,c }$
        \item\label{item:c:ccdd-}  ${\text{\makebox[\widthof{$f(d \mybot \mybot \mybot \mybot )$}][c]{$f(ccdd\mybot )$}}\,{=}\,c }$
        \item $\text{\makebox[\widthof{$f(d \mybot \mybot \mybot \mybot )$}][c]{$f(cccdd)$}}{\,{=}\,c }$
        \item $\text{\makebox[\widthof{$f(d \mybot \mybot \mybot \mybot )$}][c]{$f(ccccd)$}}{\,{=}\,c }$
        \item $\text{\makebox[\widthof{$f(d \mybot \mybot \mybot \mybot )$}][c]{$f(ccccc)$}}{\,{=}\,c }$
        \item\label{item:c:bcd--} $\text{\makebox[\widthof{$f(d \mybot \mybot \mybot \mybot )$}][c]{$f(bcd\mybot \mybot )$}}{\,{=}\,c }$
        \item\label{item:c:bccdd} ${\text{\makebox[\widthof{$f(d \mybot \mybot \mybot \mybot )$}][c]{$f(bccdd)$}}\,{=}\,c }$
        \item ${\text{\makebox[\widthof{$f(d \mybot \mybot \mybot \mybot )$}][c]{$f(bcccd)$}}\,{=}\,c }$
        \item ${\text{\makebox[\widthof{$f(d \mybot \mybot \mybot \mybot )$}][c]{$f(bcccc)$}}\,{=}\,c }$
        \item\label{item:c:bbccd} ${\text{\makebox[\widthof{$f(d \mybot \mybot \mybot \mybot )$}][c]{$f(bbccd)$}}\,{=}\,c }$
        \item ${\text{\makebox[\widthof{$f(d \mybot \mybot \mybot \mybot )$}][c]{$f(bbccc)$}}\,{=}\,c }$
        \item ${\text{\makebox[\widthof{$f(d \mybot \mybot \mybot \mybot )$}][c]{$f(bbbcc)$}}\,{=}\,b }$
        \item ${\text{\makebox[\widthof{$f(d \mybot \mybot \mybot \mybot )$}][c]{$f(bbbbc)$}}\,{=}\,b }$
        \item ${\text{\makebox[\widthof{$f(d \mybot \mybot \mybot \mybot )$}][c]{$f(bbbbb)$}}\,{=}\,b }$
        \item\label{item:c:abcd-} ${\text{\makebox[\widthof{$f(d \mybot \mybot \mybot \mybot )$}][c]{$f(abcd\mybot )$}}\,{=}\,c }$
        \item\label{item:b:abbcc} 
       ${ \text{\makebox[\widthof{$f(d \mybot \mybot \mybot \mybot )$}][c]{$f(abbcc)$}}\,{=}\,b }$
        \item ${\text{\makebox[\widthof{$f(d \mybot \mybot \mybot \mybot )$}][c]{$f(abbbc)$}}\,{=}\,b }$
        \item ${\text{\makebox[\widthof{$f(d \mybot \mybot \mybot \mybot )$}][c]{$f(abbbb)$}}\,{=}\,b }$
        \item\label{item:b:aabbc} ${\text{\makebox[\widthof{$f(d \mybot \mybot \mybot \mybot )$}][c]{$f(aabbc)$}}\,{=}\,b }$
        \item ${\text{\makebox[\widthof{$f(d \mybot \mybot \mybot \mybot )$}][c]{$f(aabbb)$}}\,{=}\,b }$
        \item ${\text{\makebox[\widthof{$f(d \mybot \mybot \mybot \mybot )$}][c]{$f(aaabb)$}}\,{=}\,b}$ 
        \item ${\text{\makebox[\widthof{$f(\mybot \mybot \mybot \mybot \mybot )$}][c]{$f(aaaab)$}}\,{=}\,a }$
        \item  ${\text{\makebox[\widthof{$f(\mybot \mybot \mybot \mybot \mybot )$}][c]{$f(aaaaa)$}}\,{=}\,a }$ 
        \item\label{item:b:-abcd} ${\text{\makebox[\widthof{$f(\mybot \mybot \mybot \mybot \mybot )$}][c]{$f(\mybot  abcd)$}}\,{=}\,b }$
        \item\label{item:b-aabb} ${\text{\makebox[\widthof{$f(\mybot \mybot \mybot \mybot \mybot )$}][c]{$f(\mybot  aabb)$}}\,{=}\,b }$
        \item ${\text{\makebox[\widthof{$f(\mybot \mybot \mybot \mybot \mybot )$}][c]{$f(\mybot  aaab)$}}\,{=}\,a }$
        \item ${\text{\makebox[\widthof{$f(\mybot \mybot \mybot \mybot \mybot )$}][c]{$f(\mybot  aaaa)$}}\,{=}\,a }$
        \item\label{item:b-a:1} ${\text{\makebox[\widthof{$f(\mybot \mybot \mybot \mybot \mybot )$}][c]{$f(\mybot \mybot  abc)$}}\,{=}\,b }$
        \item\label{item:a--aab} ${\text{\makebox[\widthof{$f(\mybot \mybot \mybot \mybot \mybot )$}][c]{$f(\mybot \mybot  aab)$}}\,{=}\,a }$
        \item ${\text{\makebox[\widthof{$f(\mybot \mybot \mybot \mybot \mybot )$}][c]{$f(\mybot \mybot  aaa)$}}\,{=}\,a }$
        \item\label{item:a-2:1b} ${\text{\makebox[\widthof{$f(\mybot \mybot \mybot \mybot \mybot )$}][c]{$f(\mybot \mybot  \mybot ab)$}}\,{=}\,a }$
        \item\label{item:a-2:1a} ${\text{\makebox[\widthof{$f(\mybot \mybot \mybot \mybot \mybot )$}][c]{$f(\mybot \mybot  \mybot aa)$}}\,{=}\,a }$
        \item\label{item:a-2:2} ${\text{\makebox[\widthof{$f(\mybot \mybot \mybot \mybot \mybot )$}][c]{$f(\mybot \mybot  \mybot \mybot a)$}}\,{=}\,a }$
        \item 
         ${f(\mybot \mybot \mybot \mybot \mybot )\,{=}\,\mybot}$ 
    \end{enumerate}
    \end{multicols}}
    \vspace{-3mm}
    \caption{}\label{abcd-rule}
\end{figure}

\begin{figure}[ht]
\centering
       \scalebox{0.95}{
        $\begin{array}{c|cccccccccccccc}
    \orbitc{c} &{\cdots} & \phantom{\, }-{ }6\phantom{\ ~ \, }& \phantom{\, }-{ }5\phantom{\ ~ \, }& \phantom{\, }-{ }4\phantom{\ ~ \, }& \phantom{\, }-{ }3\phantom{\ ~ \, }& \phantom{\, }-{ }2\phantom{\ ~ \, }& \phantom{\, }-{ }1\phantom{\ ~ \, }& \phantom{\ ~ \, }0\phantom{\ ~ \, }& \phantom{\ ~ \, }1\phantom{\ ~ \, }& \phantom{\ ~ \, }2\phantom{\ ~ \, }& \phantom{\ ~ \, }3\phantom{\ ~ \, }& \phantom{\ ~ \, }4\phantom{\ ~ \, } &\phantom{\ ~ \, } 5\phantom{\ ~ \, } &{\cdots} \\
    \hline
        {c} &{}^{\omega} \mybot  &\mybot  &\mybot  &\mybot  & \mybot  & {\textcolor{blue}{a}} & {\textcolor{purple}{b}} & {\textcolor{teal}{c}}&  {\textcolor{violet}{d}}& \mybot  &\mybot  &\mybot  &\mybot &\mybot^{\omega} \\
       {G(c)}  &{}^{\omega} \mybot&\mybot & \mybot  & {\textcolor{blue}{a}}^{\ref{item:a-2:2}} & {\textcolor{blue}{a}}^{\ref{item:a-2:1b}} & {\textcolor{purple}{b}}^{\ref{item:b-a:1}}& {\textcolor{purple}{b}}^{\ref{item:b:-abcd}} & {\textcolor{teal}{c}}^{\ref{item:c:abcd-}}&  {\textcolor{teal}{c}}^{\ref{item:c:bcd--}} & {\textcolor{violet}{d}}^{\ref{item:d:cd---}}&  {\textcolor{violet}{d}}^{\ref{item:d:d----}} &\mybot  &\mybot  &\mybot^{\omega} \\
       {G^2(c)}  &{}^{\omega} \mybot  & {\textcolor{blue}{a}}^{\ref{item:a-2:2}} &{\textcolor{blue}{a}}^{\ref{item:a-2:1a}} & {\textcolor{blue}{a}}^{\ref{item:a--aab}} & {\textcolor{purple}{b}}^{\ref{item:b-aabb}}& {\textcolor{purple}{b}}^{\ref{item:b:aabbc}}& {\textcolor{purple}{b}}^{\ref{item:b:abbcc}} &  {\textcolor{teal}{c}}^{\ref{item:c:bbccd}} &  {\textcolor{teal}{c}}^{\ref{item:c:bccdd}} &  {\textcolor{teal}{c}}^{\ref{item:c:ccdd-}} &  {\textcolor{violet}{d}}^{\ref{item:d:cdd--}} & {\textcolor{violet}{d}}^{\ref{item:d:dd---}} & {\textcolor{violet}{d}}^{\ref{item:d:d----}}& \mybot^{\omega} 
    \end{array}
    $}
    
    \caption{Evolution of $a^nb^n c^n  d^n$ where $\aneg{2}  \,{\rightarrow}\,  a ,~ \aneg{1}  \,{\rightarrow}\,  b ,~ \apos{1}  \,{\rightarrow}\,  c ,~ \apos{2}  \,{\rightarrow}\,  d $}\label{orbitevolution}
\end{figure}

\paragraph{Step-by-Step.} Let $\mathcal{A}\,{\in}\, \caclassshort{1}{2}{\Sigma}$ such that $\mathcal{L}(\mathcal{A}, I)$ for $I \,{=}\, \ppad (\{\aneg{2}  \aneg{1} \apos{1}  {\apos{2}}\})$.  
For clarity, we introduce the following renaming of the actions.
\begin{center}
    $\aneg{2}  \,{\rightarrow}\,  a ,~ \aneg{1}  \,{\rightarrow}\,  b ,~ \apos{1}  \,{\rightarrow}\,  c ,~ \apos{2}  \,{\rightarrow}\,  d $. 
\end{center}
Let's interpret the core idea through the evolution in Fig.\ref{orbitevolution}, and $f$ of $\mathcal{A}\,|_I$ is defined in Fig.\ref{abcd-rule}.
\begin{enumerate}
    \item[$\scriptstyle c_0$] \textbf{Initialization.} Let $c_0 \,{=}\, \pad(abcd) $ where $\langle0\rangle_{c_0} {=} c$.
    \item[{$\scriptstyle c_1$}] \textbf{Time Step 1 ($G(c_0)$).} We want $^{\omega}\bot a^2 b^2 c^2 d^2 \bot^{\omega}$

    \textbf{How this happens:}
    \begin{itemize}[itemindent=0cm,topsep=0.5mm,itemsep=1mm,parsep=0mm]
            \item $b$ ($\aneg{1}$), copies itself \textbf{leftward} over $a$ ($\aneg{2}$), so $a$ is \textbf{overwritten}.
            \item Similarly, $c$ ($\apos{1}$) grows rightwards over $d$ ($\apos{2}$), so $d$ is also overwritten.
            \item Hence, $a$ and $d$ must grow \textbf{two steps}:
            \begin{itemize}[itemindent=0cm,topsep=0.5mm,itemsep=1mm,parsep=0mm]
                \item One to \textbf{replace} the one that got {overwritten}.
                \item One to \textbf{expand} the block by $1$.
            \end{itemize}
    \end{itemize}
    This idea is captured by rules such as:
    \begin{itemize}[itemindent=0cm,topsep=0.5mm,itemsep=1mm,parsep=0mm]
        \item $f(\mybot \mybot \mybot \mybot \mybot) = \mybot$  : quiescent state
        \item \text{\makebox[\widthof{$f(\mybot \mybot \mybot \mybot \mybot)$}][l]{$f(\mybot a\, a \,b \,b)$}}$ = b$ : generation of $b$, override an $a$
        \item \text{\makebox[\widthof{$f(\mybot \mybot \mybot \mybot \mybot)$}][l]{$f(\mybot \mybot \mybot \mybot a )$}}$ = a$ : generation of $a$, farther left
    \end{itemize}
    The rules use radius $2$, which allows a symbol to look up to 2 cells to its left and right and decide how to evolve. This is critical to orchestrate the interactions between overlapping gliders.
\end{enumerate}
\paragraph{Inductive Growth.}
Each symbol is part of a glider pair, as it needs to propagate itself as a glider that pushes outwards:
    \begin{itemize}[itemindent=0cm,topsep=0.5mm,itemsep=1mm,parsep=0mm]
        \item \textbf{$a$'s gliders are:} $\g{a}{\mminus 2}$ and $\g{a}{\mminus 1}$ : 
        must generate a glider that reaches $2$ steps left and another of $1$ step left
        
        \item \textbf{$b$'s:} $\g{b}{\mminus 1}$ and $\g{b}{0}$ : needs to go up to $1$ step left.
        \item \textbf{$c$'s:} $\g{c}{0}$ and $\g{c}{\miniplus 1}$ : goes up to $1$ step right.
        \item \textbf{$d$'s:} $\g{d}{\miniplus 1}$ and $\g{d}{\miniplus 2}$ : Stretch to the right, goes $1$ and $2$ step right.
    \end{itemize}
    
Each glider moves and triggers local updates to emit new copies of itself and adjacent gliders in the correct directions.

\paragraph{Correctness.}
An inductive argument shows that at each time step, the CA does three things:
\begin{enumerate}[itemindent=0cm,topsep=1mm,itemsep=1mm,parsep=0mm]
    \item \emph{Preserving} existing structure
    \item \emph{Generates one more of each symbol}, maintaining the block order
    \item \emph{Recover} overwritten symbols by producing two copies from the \emph{outermost} gliders.
\end{enumerate}
The radius, $r\,{=}\,2$, allow sufficient lookahead to:
\begin{itemize}[itemindent=0cm,topsep=0.5mm,itemsep=1mm,parsep=0mm]
    \item Detect when to emit a glider or copy,
    \item Distinguish between different growing phases,
    \item Ensure that new symbols do not corrupt neighboring blocks.
\end{itemize}
This is confirmed by the detailed rule mapping (Fig.\ref{abcd-rule})

\paragraph{Overview.} This works because of careful design of gliders and their interactions:
\begin{itemize}[itemindent=0cm,topsep=0.5mm,itemsep=1mm,parsep=0mm]
    \item Each glider correspond to one \emph{unit of growth}
    \item Some overwrite others during expansion (e.g., $b$ overwrite $a$), but the overwritten symbol regenerates itself twofold - repairing the loss and expanding.
    \item This recursive glider interaction forms the basis of the desired pattern.
\end{itemize}
\section{Full Proof of \Cref{prop:421- nonk}}\label{app:full-proof-4.2.1} 
\propnonk*
Note that the language $L_r$ requires
the maintenance of multiple matching counters, which is known to be beyond the capabilities of context-free for $r\,{>}\,2$.
The construction we provide here generalizes that of Ex.\ref{exmp:anbncndn}.

\paragraph{Goal.}
Show that this language is \emph{glider-expressible} with respect to $I\,{=}\,C_\mathcal{F}(F)$ for a regular language $F$.


\paragraph{Intuition.} The key idea is to simulate the nested, non-overlapping expansion of symbol segments $\aneg{i}^{n+1}$ and $\apos{i}^{n+1}$, for all $i\in[r]$ which increase at each time step.

Our constructed CA uses gliders that propagate copies of each $\aneg{i}$ and $\apos{i}$ outward the respective direction over time for every $i\,{\in}\,[r]$.
Each $a_i$ is associated with a pair of gliders, that shifts the existing instances and
generates one new copy of its associated symbol per time step.
The velocity of each glider is definable according to its relative value's position.
Where symbols $\apos{i},\aneg{i}$ with higher $|i|$ (which are farther from the center) have higher speeds, ensuring that their gliders do not conflict with those of symbols closer to the center.
This approach scales linearly in time and maintains non-overlapping growth due to symbol priority.

\begin{proof}
Let $\mathcal{A} \,{\in}\, \caclassshort{1}{r}{\Sigma}$, and $I\,{=}\,\ppad(\{ \aneg{r} {\ldots} \aneg{1} \apos{1} {\ldots} \apos{r} \})$. 
For each $i \,{\in}\, [r]$ we define the following glider set: 
$\G_i\,{=}\,\{\g{\apos{i}}{i},\g{\apos{i}}{i-\sign(i)}\}$\footnote{That is, $\G_i \,{=}\, \{ \g{\apos{i}}{i},\g{\apos{i}}{i-1} \}$ and $ \G_{-i} \,{=}\, \{ \g{\apos{-i}}{-i},\allowbreak \g{\apos{-i}}{-i+1} \}$ for $i\,{\in}\,[r]$.}. 
The set of induced gliders is $\G {=}\bigcup_{i\in [r]} \G_i \,\cup\, \G_{-i}$.
These gliders shift the symbols $\apos{i}$ at a rate proportional to their relative location, $i$, where greater $|i|$ have greater speed.
Note that for every $i\in[-r,r]\setminus\{0\}$, the gliders of $\G_i$ satisfy the spreading state behavior (which creates one new copy at every iteration). 
That is, $\apos{1}$ is a right-spreading state (at speed $1$), $\aneg{1}$ is a left-spreading state (at speed $1$), and every other $i\in [2,r]$ is a shift-spreading state of $[i-1,i]$ and every $i\in[-r,-2]$ is a shift-spreading state of $[-i,-i+1]$.

\paragraph{Dominance Relation.}

We define a strict partial order $>$ on $\G$ as follows.
Let $\g{\apos{i}}{x}, \g{\apos{j}}{y}\,{\in}\,\G $, we write $\g{\apos{i}}{x} > \g{\apos{j}}{y}$ iff they have the same velocity direction (i.e., $\sign(i) \,{=}\, \sign(j)$), $|i|\,{<}\,|j|$ and they differ in both position and symbol, i.e., $\apos{i}\,{\neq}\,\apos{j}$ and $x\,{\neq}\, y$. 
Gliders with the same symbol or the same velocity are incomparable under $>$. 
\begin{remark}\normalfont
    The only \emph{maximal} gliders in $\G$ are of value $\aneg{1}, \apos{1}$, i.e., $\G_1$ and $\G_{-1}$, and are therefore $I$-persistent.
    All other gliders are suppressed in overlapping neighborhoods by stronger ones, and are not $I$-persistent. 
    They are, of course, derived from $\mathcal{A}|_I$ and they contribute to the local rule in the neighborhoods where they are not dominated.
\end{remark}
\paragraph{Configuration growth.}
Let $c_n$ denote the configuration at time $n$. 

We claim:
$c_n \,{=}\,^{\omega} \mybot\, \aneg{r}^{n+1} {\ldots} \aneg{1}^{n+1} \apos{1}^{n+1} {\ldots} \apos{r}^{n+1} \,\mybot^{\omega}.$
\paragraph{Base case ($n = 0$).}  
By construction $c_0\in \ppad(\{ \aneg{r} {\ldots} \aneg{1} \apos{1} {\ldots} \apos{r} \})$, thus $c_0 \,{=}\allowbreak^{\omega}\mybot \aneg{r} {\ldots} \aneg{1} \apos{1} {\ldots} \apos{r}\mybot^{\omega}$. 
Note that in the initial symbol layout, for every $i\in [r]$ we place $\aneg{i}$ at position ${-}i$, and $\apos{i}$ at position $i{-}1$.
This ensures that the gliders in $\G_{i}$ and $\G_{-i}$ start at disjoint positions and propagate outward.
\paragraph{Inductive claim.} At step $n\,{\in}\, \mathbb{N}$ we have $c_n \,{=}\,^{\omega} \mybot\, \aneg{r}^{n+1} {\ldots} \aneg{1}^{n+1} \apos{1}^{n+1} {\ldots} \apos{r}^{n+1} \,\mybot^{\omega}$, where $\langle-1,0\rangle_{c_n} = \aneg{1}\apos{1}$.
\paragraph{Inductive step:}
Assume that the claim holds for $ c_n $. Then, for any $i\in [r]$:
\begin{itemize}[itemindent=-0.1cm,topsep=0.5mm,itemsep=1mm,parsep=0mm]
    \item The gliders in $ \G_{-i} $ grow $ \aneg{i}^{n+1} $ into $ \aneg{i}^{n+2} $, leftward.
    \item Similarly, $ \G_i $ grow $ \apos{i}^{n+1} $ into $ \apos{i}^{n+2} $, rightward.
    \item This occurs via shifting and concatenation: for example, for some $i\in[r]$,
    \begin{center}
        $ \langle x, x{+}n \rangle_{c_n} \,{=}\, \aneg{i}^{n+1} \hspace{1mm} \leadsto \hspace{1mm}  \langle x{-}i, x{+}n{-}i{+}1 \rangle_{c_{n+1}} \,{=}\, \aneg{i}^{n+2},$
    \end{center}
    where $x$ is the respective index, with analogous behavior for $\apos{i}$.
\end{itemize}
Hence, at least one glider per symbol is induced at each system step (i.e., when the global function $G$ is applied).
The dominance order $>$ ensures that lower-index gliders are not suppressed.
Incomparable gliders (i.e., same symbol) do not conflict.
Thus, $ c_{n+1} \,{=}\ {}^\omega \mybot\, \aneg{r}^{n+2} {\ldots} \aneg{1}^{n+2} \apos{1}^{n+2} {\ldots} \apos{r}^{n+2} \,\mybot^\omega.$

\vspace{1mm}

\noindent
{It follows that if $f$ is the local rule obtained from applying  $ \mainproc $ on $(\G, >)$ and $\mathcal{A}=(\Sigma,r,f)$ then  $\mathcal{L}(\mathcal{A}, I)=L_r$.}
Therefore, $L_r\in \caclassglider{\Sigma}{r} $. \qed
\end{proof}

\section{Shift \& Concatenate (\Cref{lemma:shift}) --- Glider Mechanism}\label{app:shift-gliders}

\shiftlemm*

\noindent
\emph{Recall.} Defining persistent $\g{\sigma}{s}$ for every $\sigma\,{\in}\,\Sigma{\setminus}\{\mybot\}$ we get the effect that all (non-quiescent) cells are shifted by $s$ in the next iteration (Obs.\ref{obs:persistent} case II). 

\vspace{-1mm}

\begin{proof}
Define persistent gliders $\g{\sigma}{s}$ for each symbol in $w$ to shift the content of $w$ by $s$ in each iteration.
In order to create a new copy of $w$, define a second persistent glider $\g{\sigma}{s'}$ for each symbol in $w$ for $s'\,{=}\,s \,{+}\, \sign(d){\cdot} |w|$ where $\sign(d) \,{=}\, 1$ if $d$ is right and $\sign(d) \,{=}\, {-}1$ otherwise. 
These ensure that the block is replicated on the left (if $d {=} \text{left}$) or right (if $d {=} \text{right}$), without conflicting with the shifted copy.
By definition, we must have a radius $r$ that allows this modification to occur, that is; $r\,{\geq}\, \max(|s|, s+\sign(d){\cdot} k)$.
Thus, the new copy of $w$ would be generated to the left/right of the shifted part (according to $d$).
\qed
 \end{proof}   

\section{Proof of Lemma \ref{lemma-lem2}}\label{app:two-blocks}

\wnwn*

\begin{proof}
            Let $L\,{=}\,\{ w_1^n w_2^n \,|\, n \,{\in}\, \mathbb{N} \}$, $k_1\,{=}\,|w_1|$ and $k_2\,{=}\,|w_2|$.
            We write $w_1\,{=}\,\apos{0} {\ldots} \apos{k_1{-} 1}$ and $w_2\,{=}\,\bpos{0} {\ldots} \bpos{k_2{-} 1}$.
            We construct a CA of radius $r\,{=}\, \lceil \tfrac{k_1{+}k_2}{2}\rceil$, and define a set of gliders $\G$ and a strict partial order $>$ such that $(\G,>)$ induces a rule function $f$ that satisfies $\mathcal{L}(\mathcal{A}, \ppad(w_1w_2)) \,{=}\, L$.
            
\paragraph{Shift Gliders.}
            To preserve the current instance of $w_1^n w_2^n$ under iteration, we define a shift $s\,{=}\,k_1{-}r\,{=}\,\tfrac{k_1 - k_2}{2}$.
            Then for each non-quiescent symbol $\sigma\,{\in}\,\Sigma\,{\setminus}\{\mybot\}$, we define a  persistent glider $\g{\sigma}{s}$.
            These gliders shift each instance of $w_1$ and $w_2$ by $s$ positions, maintaining the existing structure.

\paragraph{New $w_1$ Gliders.}
            For new instances of $w_1$ to the left of the current block, define these gliders: $\g{\apos{i-1}}{{-}r}$ for $i\,{\in}\, [k_1]$. Each of these gliders generate a symbol $\apos{i-1}$ from a currently occupied location at distance $r$, thus generating a new block of $w_1$ to the left of the shifted configuration.
            We define their dominance as follows. Let $k_{1,i}\,{:=}\,k_1\,{-}\,i$.
            Then for every $\g{\apos{i}}{-r}$, define:
                $\g{\apos{i}}{{-}r} \,{>}\, \g{\mybot}{x}$ for all $x\,{\in}\, X_{1,i}$, where:

                \[X_{1,i} = \begin{cases}
                    {[k_{1,i}-r, r]} & \text{If } k_{1,i}\leq r,\\
                    [2r{+}1 - k_{1,i}] & \text{Otherwise}
                \end{cases} \]
 \paragraph{New $w_2$ Gliders.}
            Similarly, to produce a new instance of $w_2$ to the right of our shifted block, we define the gliders $\g{\bpos{j-1}}{r}$ for every $j\,{\in}\,[k_2]$, to generate a new instance of $w_2$ to the right at offset $r$. And we define the dominance order to be:
            $\g{\bpos{j}}{r} \,{>}\, \g{\mybot}{-x}$ for every $x \,{\in} X_{2,j}$, where

             \[X_{2,j} = \begin{cases}
                    [-r,r - k_{2,j}] & \text{If } k_{2,j}\leq r,\\
                    [2r{+}1 - k_{2,j}] & \text{Otherwise}
                \end{cases}\]
            \emph{Correctness argument.}
            We claim that the above glider system induces a CA $\mathcal{A}\,{\in}\, \caclassglider{\Sigma}{r}$ such that:
            \begin{center}
                $G^n({}^{\omega}\mybot w_1 w_2 \mybot^{\omega}) \,{=}\, {}^{\omega}\mybot w_1^{n+1} w_2^{n+1} \mybot^{\omega}.$
            \end{center}
            In each step:
            \begin{itemize}[itemindent=0cm,topsep=0.45mm,itemsep=0.735mm,parsep=0mm]
                \item Shift gliders, $\g{\sigma}{s}$, preserves our blocks by shifting $w_1^n w_2^n$ into new position,
                \item Gliders of $\g{\apos{i}}{-r}$ inject a new copy of $w_1$ to the left of the shifted block,
                \item Gliders of $\g{\bpos{j}}{r}$ inject a new copy of $w_2$ to the right of the shifted block,
                \item The dominance relation ensures that symbol-injecting gliders override any conflicting quiescent gliders in their neighborhood, and
                \item No conflicts occur between the new copies $w_1$'s gliders and $w_2$'s gliders, since they operate in disjoint directions with disjoint offsets.
            \end{itemize}
            Thus, the constructed CA evolves from the initial configuration $\ppad(w_1 w_2)$ to 
            configurations of the form$\,^{\omega} \mybot w_1^{n} w_2^{n} \mybot^{\omega}$, yielding the desired language.  
            \qed
        \end{proof}

\section{Proof of Theorem~\ref{thm:w1ntowmn-glider}}
\label{app:thm:w1ntowmn-glider}

\gliderExp*

\begin{proof}
Let $ w_1, \ldots, w_m \,{\in}\, \Sigma^+ $ be non-empty words, and define $ k_i \,{:=}\, |w_i| $. Let
\begin{center}
    $L \,{=}\, \left\{ w_1^n w_2^n \,{\ldots}\, w_m^n \;\middle|\; n \,{\in}\, \mathbb{N} \right\}, \quad r \,{=}\, \lceil\frac{1}{2} \sum_{i=1}^m k_i\rceil.$
\end{center}
Assume there exists $ \mathcal{A} \,{\in}\, \caclassshort{1}{r}{\Sigma} $ and a regular language $ F \,{\subseteq}\,\Sigma^* $ s.t. $ \mathcal{L}(\mathcal{A}, F) \,{=}\, L $.
We derive a glider system $ (\G, >) $ that induces a CA $\mathcal{A}' \,{\in}\, \caclassglider{\Sigma}{r} $ that generates the same language from the initial configuration $c_0\,{=}\,\ppad (\{w_1 w_2 {\ldots} w_m\})$.

\paragraph{Initial configuration.}
Let $ c_0 \,{=}\, {}^\omega \mybot \, w_1 w_2 \,{\ldots}\, w_m \, \mybot^\omega $. 
This initial configuration will grow into $ w_1^n w_2^n \,{\ldots}\, w_m^n $ by emitting one new copy of each $ w_i $ per time step.
\paragraph{Inductive justification.} 
We define shifts inductively to ensure minimal radius. \\[1mm]
\emph{Base case.} For $ m \,{=}\, 1 $, Lem.\ref{lemma:shift} shows that to generate $ w_1^n $, the minimal radius is of the shift size, which is $ s_1 \,{=}\, \frac{k_1}{2} $ suffices, i.e., for shifting by $s_1$ and concatenate to the left.


Note that $m =2$ is settled by Lem.\ref{lemma-lem2}; we now establish the inductive step.

\paragraph{Inductive step.} Assume the statement holds for $ m{-}1 $ words, shifts $ s_1, \ldots, s_{m-1} $ are set to fit within radius $ r_{m-1} \,{=}\, \lceil\frac{1}{2} \sum_{i=1}^{m-1} k_i \rceil$. That is, every $\sigma$ of every $w_i$ has $\g{\sigma}{s_i}$ and $\g{\sigma}{s_i - k_i}$ to have both the shift and the concatenation of the new copy (to the left in this example).
Adding $ w_m $ increases the radius to $ r_m \,{=}\, \lceil\frac{1}{2} \sum_{i=1}^m k_i\rceil$, thus, $r_m \,{=}\, r_{m-1} + \lceil\frac{1}{2} k_m \rceil$ or $r_m \,{=}\, r_{m-1} + \lceil\frac{1}{2} k_m \rceil -1$; depending on the case, without loss of generalization assume the first case.

Therefore, we adjust previous shifts to $ s_i' \,{=}\, s_i - \lceil\frac{1}{2} k_m \rceil$ so that they remain centered in the updated range. We set $s_m \,{=}\, r_m$ for the final block, pushing it as far to the right as possible. That is, $w_m$ will have for every $\sigma$ in $w_k$ the gliders $\g{\sigma}{s_m}$ (which is $\g{\sigma}{r_m}$) for maintaining existing amount and 
$\g{\sigma}{r_m-k_m}$, i.e., $\g{\sigma}{s_m - k_m}$, for the new copy concentrated to the left. 
These adjusted shifts minimize the required radius.
This way, the initial configuration structure expands by one full copy of each $ w_i $ per round.

\paragraph{Correctness.}
An inductive argument show that at each time step the CA does three things:
\begin{itemize}[itemindent=0cm,topsep=0.45mm,itemsep=0.735mm,parsep=0mm]
\item \emph{Shift} all previously emitted symbols $ w_1^n \,{\ldots}\, w_m^n $ by their respective $s_i$;
\item \emph{Emit} a new copy of each $w_i$ using the $s_i {-} k_i$ gliders;
\item Thus yielding:
$G^n(c_0) \,{=}\, {}^\omega \mybot \, w_1^{n+1} w_2^{n+1} \,{\ldots}\, w_m^{n+1} \, \mybot^\omega.$
This matches the target language $L$, and the CA operates within radius $r$ by construction.
\end{itemize}
%
Thus, the glider system defines a CA $ \mathcal{A}' \,{\in}\, \caclassglider{\Sigma}{r} $ that accepts initial configuration $ c_0$ and produces the language $L$. Hence, $L \,{\in}\, \caclassglider{\Sigma}{r},$ as required.
\qed
\end{proof}

\section{Proof of \Cref{Thm:glider-expr-pwr}}\label{app:block-repetition}

\corGliderExp*
\begin{proof}
Let $ k_i \,{=}\, |w_i| $. 
Then $r \,{:=}\, \lceil\frac{1}{2} \sum_{i=1}^m a_i \cdot k_i.\rceil$.
Assume that $ L \,{\in}\, \caclassshort{1}{r}{\Sigma} $.

\paragraph{Reduction to composed blocks.}
Let $ w_i' \,{=}\, (w_i)^{a_i}$. 
Then $L$ can be rewritten as:
\[
L \,{=}\, \left\{ (w_1')^n w_1^{b_1} (w_2')^n w_2^{b_2} \,{\ldots}\, (w_m')^n w_m^{b_m} \;\middle|\; n \,{\in}\, \mathbb{N} \right\}.
\]

\paragraph{Construction.}
We construct a glider-based CA of radius $ r $ that generates this language as follows:

\begin{itemize}[itemindent=-0.1cm,topsep=0.5mm,itemsep=1mm,parsep=0mm]
    \item Apply Thm.\ref{thm:w1ntowmn-glider} to the composed blocks $ w_1', \ldots, w_m' $. That theorem guarantees a construction where, from an initial  containing $ w_1' w_2' \,{\ldots}\, w_m' $, the configuration at step $ n $ is:
  \[
  (w_1')^{n+1} (w_2')^{n+1} \,{\ldots}\, (w_m')^{n+1}.
  \]
  \item To account for the $ b_i $ static copies of each original $ w_i $, we pre-load $ b_i $ copies of $ w_i $ into the initial configuration. These are preserved across iterations via the shift gliders but not extended, i.e., they are shifting by the respective $s_i$ of the $w_i$, but without the new copy glider.
  \item Let the initial configuration be: $c_0 \,{=}\, {}^\omega \mybot \, w_1^{a_1 + b_1} w_2^{a_2 + b_2} \,{\ldots}\, w_m^{a_m + b_m} \, \mybot^\omega$.
  \item  The construction uses the same dominance relation as in Thm.\ref{thm:w1ntowmn-glider}, where gliders responsible for symbol injection dominate quiescent alternatives.
\end{itemize}

\paragraph{Correctness.} 
Since $ L \,{\in}\, \caclassshort{1}{r}{\Sigma} $, there can be no conflict among gliders for non-quiescent symbols.
An inductive argument shows that at each time step:
\begin{itemize}[itemindent=-0.1cm,topsep=0.5mm,itemsep=1mm,parsep=0mm]
\item A single copy of each $ w_i' $ is appended via injection gliders.
\item  The $b_i$ static parts are shifted forward using standard shift gliders.
\end{itemize}
Thus, at time $n$, the configuration $c_n$ is:
\[
G^n(c_0) \,{=}\, {}^\omega \mybot \, w_1^{e_1(n+1)} w_2^{e_2(n+1)} \,{\ldots}\, w_m^{e_m(n+1)} \, \mybot^\omega.
\]
Hence, $ L \,{\in}\, \caclassglider{\Sigma}{r} $, completing the proof.
\qed
\end{proof}

\section{Example: $L\,{=}\,\{(aba)^{2n}c^{n+1} \mid n \,{\in}\, \mathbb{N}\}$}
\label{full-complex}
We consider a concrete example illustrating Thm.\ref{Thm:glider-expr-pwr} for a language with non-uniform exponents and variable word lengths.

\smallskip
\noindent
Consider the language:
$L\,{=}\,\{(aba)^{2n}c^{n+1} \mid n \,{\in}\, \mathbb{N}\}$, where $w_1\,{=}\, aba$, $w_2\,{=}\,c$. That is, $e_1(n)\,{=}\,2n$ and $e_2(n)\,{=}\,n+1$.
For each $n$, a word in $L$ has length:
$$|w_1|\cdot e_1(n) + |w_2|\cdot e_2(n)\,{=}\, 3\cdot (2n) + 1\cdot (n+1)\,{=}\, 7n+1$$

\noindent
Hence, the difference bound of $L$ is $7$, and any CA that implements $L$ must satisfy $r\,{\geq}\,\lceil{\frac{7}{2}}\rceil\,{=}\,4$.

\smallskip
\noindent
We provide a CA $\mathcal{A}\,{\in}\,\caclassshort{1}{4}{\{a,b,c,\bot\}}$ for which:
$\mathcal{L}(\mathcal{A},I)\,{=}\,L$ where $I$ is obtained by initiating the powers
with $1$. That is, $I\,{=}\,\ppad(\{(aba)^2 c^2\}){=}\,\ppad(\{abaabacc\})$.

\smallskip
We refers to $w_1' \,{=}\, (aba)^2 \,{=}\, abaaba$
and $w_2' \,{=}\, c$ as the \emph{composed blocks}.
By \Cref{Thm:glider-expr-pwr}, the language $L'\,{=}\, \{{w_1'}^n\cdot {w_2'}^{n+1} \}\,{=}\,\{{abaaba}^n\cdot {c}^{n+1} \}$ is glider-expressible at radius $r\,{=}\, \lceil\tfrac{1}{2}\cdot(|w_1'| + |w_2'|)\rceil\,{=}\, \lceil3.5\rceil\,{=}\, 4$.
This is indeed the minimal radius for this language to be CA-expressible.

\smallskip
\noindent
The CA construction is as follows:
\begin{itemize}[itemindent=0cm,topsep=0.5mm,itemsep=1mm,parsep=0mm]
    \item The initial configuration is $c_0\,{=}\, {}^{\omega}\mybot abaaba cc \mybot^{\omega}$, encoding $\ppad(w_1'^{1}, w_2'^{1+1})$.
    \item At each step, one copy of $w_1'$ and one new $c$ are added:
    \begin{itemize}[itemindent=0cm,topsep=0.5mm,itemsep=1mm,parsep=0mm]
        \item The $w_1'$ block is generated via a set of gliders that shift it $4$ steps to the left (i.e., $s\,{=}\,-4$) and concatenate it to the right.
        \item The $w_2$ block (the $c$) grows to the right, via gliders that append one new $c$ per step, that is, $\g{c}{+2}$ and $\g{c}{+3}$.
    \end{itemize}
\end{itemize}
This yields configurations of the form:
\begin{center}
    $G^n(c_0) \,{=}\, {}^{\omega} \mybot (abaaba)^{n+1} c^{n+2} \mybot^{\omega}$
\end{center}
which corresponds to the language $L$ with the expected exponents.
This construction satisfies the conditions of Thm.\ref{Thm:glider-expr-pwr}, using radius $r\,{=}\,4$.

\end{document}